\documentclass[pdflatex,sn-mathphys-num]{sn-jnl}

\usepackage{graphicx}%
\usepackage{multirow}%
\usepackage{amsmath,amssymb,amsfonts}%
\usepackage{amsthm}%
\usepackage{mathrsfs}%
\usepackage[title]{appendix}%
\usepackage{xcolor}%
\usepackage{textcomp}%
\usepackage{manyfoot}%
\usepackage{booktabs}%
\usepackage{algorithm}%
\usepackage{algorithmicx}%
\usepackage{algpseudocode}%
\usepackage{listings}%
\usepackage[capitalize]{cleveref}
\usepackage{subcaption}
\usepackage{longtable}
\usepackage{amsbsy}
\usepackage{dsfont}
\usepackage{bbding}

\theoremstyle{remark}
\newtheorem{remark}{Remark}

\newcommand{\Jcal}[2]{\mathcal{J}^{#1}_{#2}}

\def\~{\tilde}
\def\N{\mathbb N}
\def\R{\mathbb R}
\def\K{{\vec K}}
\newcommand{\lpar}{(}
\newcommand{\rpar}{)}
\newcommand{\lbr}{\left[}
\newcommand{\rbr}{\right]}
\newcommand{\pro}{\mathbb{P}}
\newcommand{\E}{\mathbb{E}}
\newcommand{\norm}[1]{\left\|#1\right\|}
\newcommand{\vi}{{\vec i}}
\newcommand{\vj}{{\vec j}}

\counterwithout{equation}{section} 
\newcommand{\A}{\mathcal A}

\newcommand{\one}{\mathds{1}}

\newcommand{\Zcal}{\mathcal{Z}}

\theoremstyle{thmstyleone}%
\newtheorem{theorem}{Theorem}
\newtheorem{lemma}{Lemma}
\newtheorem{definition}{Definition}%

\usepackage{amsthm}
\makeatletter
\newtheorem*{rep@theorem}{\rep@title}
\newcommand{\newreptheorem}[2]{%
\newenvironment{rep#1}[1]{%
 \def\rep@title{#2 \ref{##1}}%
 \begin{rep@theorem}}%
 {\end{rep@theorem}}}
\makeatother

\newreptheorem{theorem}{Theorem}
\newreptheorem{lemma}{Lemma}

\renewcommand{\orcid}[1]{\href{https://orcid.org/#1}{\orcidlogo}}

\begin{document}

\title[MRJ Scheduling with Continuous Resource Requirement]{Throughput-Optimal Multiresource-Job Scheduling with Continuous Requirement Distribution}


\author*[1]{\fnm{Heyuan} \sur{Yao} \orcid{0009-0007-6587-3485}}\email{heyuanyao2029@u.northwestern.edu}

\author[1]{\fnm{Willow} \sur{Kowalik}}\email{j8k9q7@u.northwestern.edu}

\author[1]{\fnm{Izzy} \sur{Grosof} \orcid{0000-0001-6205-8652}}\email{izzy.grosof@northwestern.edu}


\affil*[1]{\orgdiv{Department of Industrial Engineering and Management Sciences}, \orgname{Northwestern University}, \orgaddress{\street{2145 Sheridan Rd}, \city{Evanston}, \postcode{60208}, \state{IL}, \country{USA}}}



\abstract{
Modern computing systems process jobs with resource requirements such as CPU and memory, which are described by multiresource jobs (MRJ) queueing models. In practice, job resource requirements are spread out over so many values, that it is rare to see the same value twice. This pattern is best modeled by a continuous distribution of requirement values. However, the existing theoretical work on stability or throughput-optimality focuses on queueing models with class-based resource requirements. In class-based models, the number of distinct resource requirements must be small to demonstrate strong empirical performance, making them a poor match for these practical systems.

We introduce the first throughput-optimal family of scheduling policies for the continuous MRJ model, with both preemptive and nonpreemptive variants. We further introduce several efficient policy families, which remain throughput-optimal while considerably improving computational efficiency, under some distributional assumptions. We use a discretization approach, where we choose the discretization granularity based on the system load and the distribution of resource requirements. 
We validate the real-world applicability of our policies by comparing them against existing index-based policies on parametrized distributions and on datacenter trace data from the Google Borg scheduler, demonstrating state-of-the-art performance.
}




\maketitle

\section{Introduction}
\label{sec: Introduction}

Modern computing systems, such as cloud platforms, datacenters, and personal machines, typically possess multiple kinds of computational resources (e.g., CPU, GPU, memory). These systems can be modeled as queues where jobs require certain amounts of these resources, in order to be served. We refer to these jobs as multiresource jobs.\footnote{When the system is equipped with only one resource type, we refer to these jobs as multiserver jobs.} The system is able to serve any set of multiresource jobs simultaneously, provided that the total requirement of each resource does not exceed the corresponding resource capacity.

A scheduling policy specifies which set of multiresource jobs should be served at any time. Since the 1980s, many policies have been extensively studied \cite{green1980queueing,brill1984queues,tassiulas1993dynamic,coffman2001bandwidth,stolyar2004maxweight,maguluri2012stochastic,rumyantsev2017stability,psychas2018randomized,grosof2022wcfs,grosof2022optimal,grosof2023reset,hong2022sharp,CHEN2026102525}, and their analyses include strong results such as system stability and throughput-optimality in their respective settings.

However, existing policies, such as MaxWeight \cite{stolyar2004maxweight,maguluri2012stochastic}, Randomized Multi-Resource Scheduling (RMS) \cite{psychas2017non}, Markovian Service Rate (MSR) \cite{chen2024analyzing}, are restricted to settings with a finite set of possible resource requirements, which are known as ``job classes". For instance, MaxWeight weights jobs by how many jobs of identical resource requirements are present. The larger the number of job classes, the less the empirical performance corresponds to the theoretical results. In particular, the number of jobs present must be of at least the order of the number of job classes for empirical performance to correspond to the theoretical results. These queue lengths are too large for practical purposes. In datacenter scheduling traces, hundreds of thousands of unique resource requirements are commonly seen \cite{tirmazi2020borg}.
Theoretical results, which only kick in above hundreds of thousands jobs, are not useful.

In real-world scheduling environments, for example, the first million jobs of the Borg trace \cite{tirmazi2020borg} have 3,668 unique memory requirements, 536,229 unique CPU requirements, and 949,217 unique (CPU, memory) pairs. As a result, the existing class-based policies used in the literature are not suitable for such practical settings. Seeing this mismatch, we introduce the \emph{continuous MRJ model}, in which the distribution of the resource requirement is continuous. While real systems do not literally have continuous resource requirements, the continuous MRJ model captures the behavior of rarely repeating resource requirements. In this model, we ask a question:
\begin{quote}
\emph{Does there exist a throughput-optimal family of scheduling policies in the continuous MRJ setting?}
\end{quote}

We resolve this question by developing the first throughput-optimal family of scheduling policies, the $\K$-discretized MaxWeight ($\K$-MW) family. $\K$-MW uses discretization intentionally, rather than assuming that the resource requirements are inherently well-discretized. Our proof builds off the existing class-based MaxWeight proof, with additional techniques to handle the overhead of discretization, which must be controlled and mitigated.

Many applications involve nonpreemptible jobs or require substantial preemption overhead. We therefore introduce the first nonpreemptive throughput-optimal family of scheduling policies: $\K$-discretized nonpreemptive Markovian Service Rate ($\K$-nMSR). For $\K$-nMSR, we use the same discretization approach as $\K$-MW and construct the policy using the offline precomputation method for the MSR policies introduced in \citet{chen2024analyzing}.

In addition, computational efficiency is essential for both the online optimization in MaxWeight-type policies and the offline precomputation in nMSR-type policies. For instance, when using standard class-based MaxWeight in systems with numerous classes and small expected service durations, the exponentially large number of possible service options may cause the system to spend more time updating service options than actually serving jobs. To achieve computational efficiency, we must avoid handling an extremely large candidate set. We introduce a new group of policies that restrict attention to smaller candidate sets, where we only consider a small fraction of the possible service options, during MaxWeight and nMSR's optimization steps. We call these smaller candidate sets ``efficient sets".

We start by considering only the service options that form extreme vertices of the convex hull of the set of all service options. By doing so, for example, we reduce the candidate set by around $200\times$ when $K = 60$ in the single-resource setting. Next, in the settings with symmetrical resource requirement distribution, our 2-Job efficient set reduces the number of service options to a quantity proportional to the number of discretization classes. In the single-resource setting with decreasing requirement density, our 2-Bucket efficient set likewise considers a number of service options proportional to the number of discretization classes. Our policies built on these three efficient sets are each throughput-optimal under corresponding distributional assumptions on the resource requirements. 

We use stimulation to extensively verify the performance, both throughput and mean response time, in both parameterized settings and a public Borg trace. Our simulation results verify the throughput optimality of the $K$-MW family and the Efficient $K$-MW families on a variety of requirement distributions. Moreover, both on these distributions and in experiments based on job requirements extracted from the Google Borg trace, our theoretically-motivated policies outperform in mean response time existing index-based policies: FCFS, First-Fit, Best-Fit, and Least-Server-First (LSF). Note that LSF is equivalent to naive MaxWeight in the continuous MSJ setting, because every job in its own class has equal weight. Naive MaxWeight then always serves most jobs, and packs the jobs in increasing order of their resource requirement.

This paper is organized as follows:
\begin{itemize}
    \item \cref{sec: prior work} reviews related prior work on multiserver and multiresource scheduling.
    \item \cref{sec: model} introduces the multiresource model under the continuous resource requirement setting, the $\K$-discretization scheme, and our proposed policies.
    \item \cref{sec: main results} presents our main theorems, including throughput-optimality results for each family of policies.
    \item \cref{sec: proofs} provides detailed proofs of the main theorems.
    \item \cref{sec: empirical results} compares our policies with existing policies via simulations.
\end{itemize}

\section{Prior Work}
\label{sec: prior work}

In this section, we review real-world scheduling practices in datacenters and supercomputers, as well as the theoretical development of scheduling policies for multiresource-job systems.

\subsection{Scheduling in Large-Scale Computing Settings}
\label{ssec: Scheduling in Large-Scale Computing Settings}
Traditional one-server-per-job queueing systems are not suitable for modeling many modern computing environments. A more accurate representation is provided by the multiserver-job (MSJ) system, where a job requires some amount of a single resource, and the multiresource-job (MRJ) system, where a job requires some amount of multiple resources. LLM handling infrastructure, such as within Microsoft Azure \cite{patel2024splitwise}, handles inference tasks, which vary by having a different number of tokens in different requests, corresponding to an MSJ model. The number of tokens implicitly corresponds to GPU and memory requirements.

Similarly, traditional datacenters process jobs that use varying amounts of resources in multiple categories, such as CPU, memory, and more, corresponding to an MRJ model. These multiresource requirements are present at the Google Borg \cite{reiss2012heterogeneity,clusterdata:Verma2015,tirmazi2020borg,clusterdata:Wilkes2020,clusterdata:Wilkes2020a,clusterdata:Wilkes2011,clusterdata:Reiss2011} and Google Omega \cite{schwarzkopf2013omega} systems, as well as Alibaba's Sigma and Fuxi systems \cite{guo2019limits,li2019deepjs,wu2019aladdin}. Finally, high-performance computing jobs submitted to supercomputers require a pre-specified number of nodes, for instance, at Argonne's Aurora \cite{allen2025aurora} and Oak Ridge's Frontier \cite{budiardja2023ready,atchley2023frontier} supercomputers, corresponding to an MSJ model.

Furthermore, the distribution of these resource requirements has extremely large support, subdivided smaller than any useful measure of precision. For instance, the traces available from Google Borg \cite{clusterdata:Wilkes2020,clusterdata:Wilkes2020a,clusterdata:Tirmazi2020} demonstrate that CPU and memory requests are specified as floats with at least 13 digits of precision, and have nearly no repeated pairs: around 95\% unique pairs over the first 1 million requirements. The RAM usage in Google Omega is recorded in bytes and is therefore specified down to the individual byte, e.g., 1074000000, 82587) \cite{schwarzkopf2013omega}. In Microsoft Azure, although token requirements are discretely represented, traces from multiple LLM inference services over a single day include $2339$ distinct context-token counts and $623$ distinct generated-token counts \cite{patel2024splitwise}. In all of these scenarios, class-based scheduling, where the ``class" refers to a set of jobs with identical resource requirements, becomes unsuitable, as two jobs with only minor differences in some resource requirements would be assigned to distinct classes. This observation motivates the need to consider scheduling policies under the settings where resource requirements are continuous or of large support.

\subsection{Multiserver-job and Multiresource-job Scheduling} 
\label{ssec: Multiserver-job Scheduling and Multiresource-job Scheduling}

We now review prior work on scheduling policies for the MSJ and MRJ settings. Some policies are designed primarily for the MSJ system and do not extend to the MRJ system.

The policies proposed in prior work can be broadly divided into two categories: index-based policies and class-based policies. Specifically, a job class refers to a set of jobs with identical resource requirements. Class-based policies schedule based on the number of jobs in each class that are present in the system at a given time. In other words, class-based policies rely on job-type information when making scheduling decisions. This implicitly assumes that job classes are discrete and finite, and that the number of possible classes is smaller than the number of jobs typically present in the system, so that the queue-length information for each class is meaningful. In contrast, index-based policies do not classify jobs into finitely many types. Instead, an index-based policy orders the jobs in the system according to a certain rule, iterates over all jobs or a subset of jobs, and decides whether to serve each job in turn. \cref{table: list of policies} summarizes the properties of some representative index-based and class-based policies in the prior work compared with our policies introduced in \cref{sec: model}. Notably, all prior theoretical works prove analytical results only in settings where resource requirements are finitely supported. To the best of our knowledge, no theoretical results have been established for continuous-resource settings, which are the focus of this paper.

\begin{table}[ht]
\begin{tabular}{|c|c|c|c|c|}
\hline
Policies                                                         & Nonpreemptive & MRJ Setting & Throughput-Optimal                                                      & Continuous Resource                                                         \\ \hline
FCFS                                                             & \Checkmark    & \Checkmark  & \XSolid                                                                 & \begin{tabular}[c]{@{}c@{}}Analysis only \\ in discrete setting\end{tabular} \\ \hline
First-Fit & \Checkmark    & \Checkmark  & \begin{tabular}[c]{@{}c@{}}\Checkmark \\ under assumptions \cite{coffman2001bandwidth} \end{tabular} & \begin{tabular}[c]{@{}c@{}}Analysis only \\ in discrete setting\end{tabular} \\ \hline
Best-Fit  & \Checkmark    & \Checkmark  & \begin{tabular}[c]{@{}c@{}}\Checkmark \\ under assumptions \cite{coffman2001bandwidth} \end{tabular} & \begin{tabular}[c]{@{}c@{}}Analysis only \\ in discrete setting\end{tabular} \\ \hline
LSF  & \Checkmark    & \Checkmark  & \XSolid & \begin{tabular}[c]{@{}c@{}}Analysis only \\ in discrete setting\end{tabular} \\ \hline
ServerFilling & \Checkmark    & \XSolid     & \begin{tabular}[c]{@{}c@{}}\Checkmark \\ under assumptions \cite{grosof2022wcfs} \end{tabular} & \begin{tabular}[c]{@{}c@{}}Analysis only \\ in discrete setting\end{tabular}                                                                    \\ \hline
\begin{tabular}[c]{@{}c@{}}ServerFilling-\\ SRPT\end{tabular} & \XSolid    & \XSolid     & \begin{tabular}[c]{@{}c@{}}\Checkmark \\ under assumptions \cite{grosof2022optimal} \end{tabular} & \begin{tabular}[c]{@{}c@{}}Analysis only \\ in discrete setting\end{tabular}                                                                       \\ \hline
MaxWeight                                                        & \XSolid       & \Checkmark  & \Checkmark                                                              & \XSolid                                                                      \\ \hline
RMS                                                              & \Checkmark       & \Checkmark     & \Checkmark                                                              & \XSolid                                                                      \\ \hline
MSR                                                              & \Checkmark    & \Checkmark  & \Checkmark                                                              & \XSolid                                                                      \\ \hline
MSFQ                                                             & \Checkmark    & \XSolid     & \Checkmark                                                              & \XSolid                                                                      \\ \hline
 \textbf{$\K$-MW    }                                                         & \pmb{\XSolid}    & \pmb{\Checkmark}     & \pmb{\Checkmark}                                                              & \pmb{\Checkmark}                                                                      \\ \hline
  \textbf{$\K$-MSR   }                                                          & \pmb{\Checkmark}    & \pmb{\Checkmark}     & \pmb{\Checkmark}          & \pmb{\Checkmark}      \\ \hline
\end{tabular}
\caption{Summary of the settings and the throughput-optimal condition of existing index-based and class-based policies, and our $\K$-discretized MaxWeight ($\K$-MW) and $\K$-discretized nonpreemptive Markovian Service Rate ($\K$-nMSR) families of policies.}
\label{table: list of policies}
\end{table}

We now introduce some classical index-based policies. As a representative index-based policy, First-Come-First-Served (FCFS) has been studied since the 1980s. Stability and mean response time analyses in the two-server setting were conducted in \cite{green1980queueing,brill1984queues}. In the integer-valued requirement setting, the stability region of FCFS was characterized in \cite{rumyantsev2017stability}, and the first mean response time analysis was conducted in \cite{grosof2023reset}. Other well-known index-based policies include First-Fit, Best-Fit, and Least-Server-First (LSF), each of which uses Backfilling to greedily pack jobs according to some scanning order. We elaborate on the definitions of these three policies in Appendix~\ref{sec: Definitions of Some Index-Based Policies}. First-Fit and Best-Fit are throughput-optimal if the requirement distribution is symmetric \cite{coffman2001bandwidth}. However, they generally fail to be throughput-optimal in more general settings \cite{maguluri2012stochastic}. \citet{hong2022sharp} studied the mean response time of LSF in a low-load scaling regime, showing that LSF achieves the same asymptotic growth rate as the optimal mean response time bound. However, LSF is generally not throughput-optimal because it prioritizes small jobs and fails to fully utilize system capacity. 
ServerFilling \cite{grosof2022wcfs} and ServerFilling-SRPT \cite{grosof2022optimal} are recently introduced near-index policies, with strong mean response bounds in the power-of-two MSJ requirement setting. Although most of these index-based policies can naturally be applied to the continuous MSJ or MRJ framework, their theoretical analysis has focused only on discrete settings, where job requirements have finite support. Moreover, none of them is throughput-optimal in the general resource requirement distribution setting.

We next move to some well-known class-based policies. One of the earliest preemptive class-based policies in the literature is MaxWeight \cite{tassiulas1993dynamic,stolyar2004maxweight}, which is also the earliest policy proven to be throughput-optimal in the MRJ setting \cite{maguluri2012stochastic}. Despite its throughput-optimality and low empirical mean response time, MaxWeight repeatedly solves an online optimization problem and chooses service options whenever the system state changes, leading to frequent preemption and computational inefficiency. The first nonpreemptive policy proven throughput-optimal in the MRJ setting, Randomized Multi-Resource Scheduling (RMS, also known as Randomized-Timer) \cite{psychas2018randomized}, addresses these issues but often results in a higher empirical mean response time than MaxWeight. A family of policies called Markovian Service Rate (MSR) was recently introduced in \cite{chen2024analyzing}. MSR can accommodate both preemptive and nonpreemptive settings and is throughput-optimal with bounded mean response time. \citet{chen2024analyzing} showed that nonpreemptive MSR (nMSR) achieves better empirical mean response times than RMS, especially under high load. However, nMSR has an issue: it ignores the queue-length state when switching service options, resulting in still subpar mean response time. To reduce the mean response time of nMSR in the MSJ setting, \citet{CHEN2026102525} proposed Most Servers First with Quickswap (MSFQ). MSFQ is throughput-optimal in the one-or-all case and empirically achieves better mean response time and weighted mean response time than nMSR in the general integer-valued resource requirement case. Many class-based policies achieve throughput-optimality, in contrast to index-based policies.

Class-based policies are not immediately applicable to the continuous MSJ or MRJ settings, because job requirements are continuously distributed, and the system almost never contains multiple jobs with exactly identical requirements at any given time. We will build on existing class-based policies (MaxWeight and nMSR) by using a discretization scheme, to achieve throughput-optimality in the continuous setting. For a boost to empirical performance, we also augment our discretization-based policies with an index-based Backfilling approach.

\section{Model}
\label{sec: model}

In this section, we define single- and multi-resource queueing models with continuously distributed resources (the continuous MSJ and MRJ settings). We first present the model and notation (\cref{ssec: notations,ssec: Multiresource-Job Queueing System with Continuous Resource Requirement}), then introduce our discretization-based policies in \cref{ssec: New Scheduling Policies: The Discretization-Based Policies} and their counterparts with Backfilling in \cref{ssec: Backfilling}. 

\subsection{Basic Notation}
\label{ssec: notations}
We let $\N$ be the set of all nonnegative integers and $[K]: = \{0,1,...,K\}$ be the set of all nonnegative integers not greater than $K$. In an Euclidean space $\R^d$, we let an arrowed letter denote a deterministic vector (e.g., $\vec a$). We let $\vec {0} = \lpar0,...,0\rpar$ and $\vec {1} = \lpar1,...,1\rpar$. 
For two vectors $\vec a ,\vec b \in \R^d$, we say $\vec a \preceq \vec b$ if $a_i \leq b_i$ for all $1\leq i \leq d$, and let $\vec b- \vec a = \lpar b_1-a_1,...,b_d-a_d\rpar$ denote the difference of two vectors. Given a set $C \subseteq \R^d$, let $ConvH\lpar C\rpar$ denote the convex hull of $C$ and $int\lpar C\rpar$ denote the interior of $C$. A function $f$ with domain $D \subseteq \R^d$ is said to have central symmetry (or point reflection symmetry) with respect to a point $\vec a \in \R^d$ if for any $v \in D$, $f(v) = f(2\vec a-v)$.

For a function $f: \R^d \rightarrow \R$, we let $supp\lpar f\rpar$ denote its support. For a probability distribution $\pi$, we let $\mathcal{L}\lpar \pi| A\rpar$ denote the conditional distribution (law) on some measurable set $A$.

For an $\vec{K} = \lpar K_1,...,K_d\rpar \in \N^d$, we let $\Gamma_{\vec K}: = \{\vi \in \N^d:\vec {1} \preceq \vec {i} \preceq \vec {K}\}$ and let $\prod \K : = |\Gamma_\K| = \prod_{l=1}^d K_l$. 

For a vector $Z \in \R^{\Gamma_\K}$ with multi-dimensional indices, we let $Z^{\lpar \vec {i}\rpar}$ denote the $\vec{i}$-th element of $Z$, for $\vi \in \Gamma_\K$.

In our paper, we always let $M\in \N^{\Gamma_\K}$ denote a service option, where $M^{(\vi)}$ is the number of type-$\vi$ jobs (See  \cref{ssec: Auxiliary Notations}) served by $M$. However, when $\K$ is (entry-wise) large and the vector is sparse, which happens frequently in our setting, it is inconvenient to explicitly write out $M$. Hence, we provide two equivalent notations for $M$.

\begin{definition}[Alternative Notation I for $M\in \N^{\Gamma_\K}$]
    \label{def: alternative notation I for M}
    Any $M\in \N^{\Gamma_\K}$ can be equivalently expressed as an unordered list of non-zero elements with repetitions. To be specific, if $\{\vi_1, \vi_2,..., \vi_n\} := \{\vi \in \Gamma_\K: M^{(\vi)} \neq 0\}$ are the non-zero entries of $M$, then we can equivalently write $M$ such that
    \begin{equation}
        \label{equ: alternative notation I for M}
        M = \left\{ \underbrace{\vi_1,...,\vi_1}_{M^{(\vi_1)} \text{'s}}, \underbrace{\vi_2,...,\vi_2}_{M^{(\vi_2)} \text{'s}},...,\underbrace{\vi_n,...,\vi_n}_{M^{(\vi_n)} \text{'s}}  \right\}.
    \end{equation}
    Note that the order of elements does not matter for the uniqueness of $M$.
\end{definition}

\begin{definition}[Alternative Notation II for $M\in \N^{\Gamma_\K}$]
    \label{def: alternative notation II for M}
    Any $M\in \N^{\Gamma_\K}$ can be equivalently expressed as a list of pairs of unique elements and their multiplicities.
    To be specific, if $\{\vi_1, \vi_2,..., \vi_n\} := \{\vi \in \Gamma_\K: M^{(\vi)} \neq 0\}$ are the non-zero entries of $M$, then we can equivalently write $M$ such that
    \begin{equation}
        \label{equ: alternative notation II for M}
        M = \left\{(\vi_1,M^{(\vi_1)}), (\vi_1,M^{(\vi_1)}),..., (\vi_n,M^{(\vi_n)})  \right\}.
    \end{equation}
    Note that the order of pairs does not matter for the uniqueness of $M$.
\end{definition}

We use $M = (1,0,2,0,1,0,0,0) \in \N^8$ as an example. Based on the alternation notation I given by \cref{def: alternative notation I for M}, an equivalent expression of $M$ is $\{1,3,3,5\}$. Based on the alternation notation II given by \cref{def: alternative notation II for M}, an equivalent expression of $M$ is $\{(1,1),(3,2),(5,1)\}$.

For the list of all notation, see Appendix~\ref{sec: List of Notations}.

\subsection{Multiresource-Job Queueing System with Continuous Resource Requirements}
\label{ssec: Multiresource-Job Queueing System with Continuous Resource Requirement}

We study a model where a server has $d$ distinct resources available, and each job requires some amount of each resource. Without loss of generality, we normalize the server's available capacity for each resource to be $1$. A set of jobs can be served simultaneously whenever, for each resource type, the total requirement of these jobs does not exceed the capacity $1$.

Job arrivals follow a Poisson arrival process, with arrival rate $\lambda$. Resource requirements are sampled i.i.d. from the underlying distribution with random variable $V$, which takes values in $( 0,1]^d$. In addition to a resource requirement, each job has a duration demand, sampled i.i.d. from a distribution with random variable $D$. The states of the system are lists of (requirement, duration) pairs, ordered according to their arrival times. A scheduling policy selects a set of jobs to be served simultaneously, and is assumed to be a Markovian policy, which depends only on the current state. We consider both the preempt-resume and nonpreemptive settings. When we refer to a queueing system as stable, formally we mean that the empty state is positive recurrent.

In this paper, we do not assume the resource is pre-discretized. This is in contrast to prior work, which typically assumes the job resource requirements distribution is discrete with finite support. Our assumption is more useful for real-world cloud and cluster scheduling. Specifically, resource requirement distribution $V$ is characterized by the probability density function (p.d.f.) $f_V: (0,1]^d \rightarrow \R$. For simplicity, we focus on the case where the duration $D$ follows the exponential distribution $Exp(1)$, and the scheduling policy does not observe the remaining duration demand. However, we believe our policy should perform well under the general duration settings.

Under these assumptions, the state space simplifies to the set of current resource requirements, denoted by
\[\mathcal S : = \left\{x = \left (\vec v_n\right )_{1\leq n \leq N}\subseteq ( 0,1]^d: N \in \N \right\}, \]
where the smaller index $n$ means an earlier arrival time. A Markovian policy can then be defined as follows.

\begin{definition}
    \label{def: Markovian policy}
    A Markovian policy $\pi$ maps each state $x$ to a subset $x' \subset x$, such that
    \begin{equation}
    \label{equ: any policy formula}
        \sum_{\ell = 1}^{|x'|} x'_\ell \preceq \vec 1
    \end{equation}
    The subset $x'$ refers to the set of jobs the server serves simultaneously. We define $\mathscr S$ to be the set of all schedulable service options $x'$.
\end{definition}

Thereafter, all policies we will discuss are Markovian policies. To study the performance of a policy, we apply the drift method by specifying a Lyapunov function (test function) and analyzing its drift. Specifically, a Lyapunov function $L: \mathcal{S} \rightarrow \R^+_0$ maps each state of to a nonnegative number. The drift of a Lyapunov function $L$ under some policy $\pi$, $\Delta_\pi \circ L(x)$, is then defined as the instantaneous change rate of the $L$ at the state $x$, 
\begin{equation}
    \label{equ: drift of V}
    \Delta_\pi \circ L(x)= \lim\limits_{t\downarrow 0} \frac{1}{t} \E \lbr L(x( t )) - L(x(0))\,|\, x(0) = x\rbr.
\end{equation}

\subsection{New Discretization-Based Scheduling Policies}
\label{ssec: New Scheduling Policies: The Discretization-Based Policies}

In \cref{ssec: Auxiliary Notations}, we first introduce our method of discretizing the resource requirements, and provide corresponding notation. We then introduce our discretization-based scheduling policies in the $\K$-discretization system: $\K$-discretized MaxWeight (\cref{sssec: K-discretized MaxWeight Scheduling Policy}), the $\K$-discretized efficient MaxWeight family (\cref{sssec: Efficient discretized MaxWeight}), $\K$-discretized nonpreemptive Markovian Service Rate (\cref{sssec: K-discretized Markovian Service Rate}), and $\K$-discretized efficient nonpreemptive Markovian Service Rate (\cref{sssec:KE nMSR}). 

\subsubsection{The $\K$-Discretization System} 
\label{ssec: Auxiliary Notations}

We begin with the single-resource queueing model, where $ V \in ( 0,1]$. After fixing the discretization parameter $K\in \N$, we equally partition the resource range $( 0,1]$ such that 
\begin{equation}
\label{equ: K-partition}
    ( 0,1] = \bigcup_{k=1}^K   I_k,\quad \text{where } I_k:=\left( \frac{k-1}{K}, \frac{k}{K}\right ],  \text{ for } k = 1,..,K,
\end{equation}
where $I_k$ is the interval in the continuous requirement space corresponding to a single discrete requirement. We also call each interval $I_k$ a ``bucket". Given a job with resource requirement $v$ in the set $I_k$, where $k = \left\lceil K v \right\rceil$, we say the job is a type-$k$ job. At time $t$, let $q( t ) = \lpar q_1( t ),...,  q_K( t ) \rpar \in \N^K$ such that $q_k( t )$ denotes the number of type-$k$ jobs in the system. We call a policy ``$K$-discretized" if the policy uses only the information of the job types. Therefore, under a $K$-discretized policy, the process $\{q( t )\}_{t\geq 0}$ forms a continuous-time Markov chain (CTMC). 

Let \[M( t ) = (M^{( 1 )}( t ),...,M^{( K )}( t )) \in [K]^K\] denote the service option the server chooses at time $t$, where $M^{( k )}( t )$ denotes the number of type-$k$ jobs in service, for each $k = 1,..., K$. Restricted by the capacity of the server, our service option should satisfy that $M( t ) \in C_K$, the set of all admissible service options, which is defined as 
\begin{equation}
    \label{equ: C_K a-dim}
    C_K := \{ M \in [K]^K: \langle M, (1,...,K)\rangle = \sum\limits_{k=1}^K k M^{( k )}  \leq K \}. 
\end{equation}

The $d$-resource setting follows similarly. By fixing some $\vec {K}: = (K_1,...,K_d )\in \N^d$, recall that we let $\Gamma_{\vec K}: = \{i \in \N^d:\vec {1} \preceq \vec {i} \preceq \vec {K}\}$ denote the set of all job types. We partition the set $( 0,1]^d$ such that
\begin{equation}
    \label{equ: partition in d-dim}
    ( 0,1]^d = \bigcup_{ \vec  i \in \Gamma_{\vec  K}} I_{\vec i} \quad \text{where }
    I_{\vec i}:= \left (\frac{i_1-1}{K_1}, \frac{i_1}{K_1}\right]\times...\times \left (\frac{i_d-1}{K_d}, \frac{i_d}{K_d} \right] \text{ for all } \vec i \in \Gamma_{\vec K}.
\end{equation}
Each $I_{\vec i}$ is a $d$-dimensional box in the continuous requirement space, corresponding to a single discrete requirement. Given a job with resource requirements $\vec v= (v_1,...,v_d)$ in the set $I_{\vec {i}}$, where $\vec {i} = \left(\left\lceil K_1 v_1 \right\rceil,..., \left\lceil K_d v_d \right\rceil\right )$, we say the job is a type-$\vec i$ job. 
Similarly to the single-resource setting, we let $q( t ) = \lpar q_{\vec i}( t )\rpar_{\vec i \in \Gamma_ {\vec K}}$, where $q_{\vec i}( t )$ denotes the number of type-$\vec i$ jobs in the system at time $t$. We similarly define a $\K$-discretized policy to only use the information of the job types. Specifically, we let $M(t) \in \N^{\Gamma_\K}$ denote the service option at time $t$. Note that $M(t) \in C_\K$, the set of all schedulable options in $\K$-discretized system, such that 
\begin{equation}
    \label{equ: C_K d-dim}
     C_{\vec K} := \{ M \in \left[K_1\right]\times ...\times \left[K_d\right]: \sum_{\vec {i} \in \Gamma_{\vec {K}} } i_l M^{(\vec {i})}  \leq K_l \; \text{for all}\; 1\leq l \leq d\}, 
\end{equation}
where $C_{\vec K}$ is the set of all admissible service options under $\vec{K}$-discretization, and $M^{(\vec i)}( t )$ denotes the number of type-$\vec i$ jobs in that are served at time $t$.

We now provide some definitions on arrival and service measures, before we move to introducing our novel discretization-based policies. We recall that the system has the arrival rate $\lambda$ and the job resource requirement distribution $V\in ( 0,1]^d$ with density $f_V:( 0,1]^d \rightarrow \R$. 

\begin{definition}[($\K$-Discretized) Arrival Measures]
    \label{def: arrival measure}
    Given the arrival rate $\lambda$ and the resource requirement density $f_V$, the \textbf{arrival measure} $\eta^A$ on the space $( 0,1]^d$ is given by 
    \begin{equation}
        \label{equ: def of mu^A}
        \eta^A(dx) : = \lambda f_V(dx).
    \end{equation}
    Moreover, for any fixed $\K\in \N^d$, the \textbf{$\vec {K}$-discretized arrival measure} ($\vec {K}$-arrival measure) $\eta^A_{\vec K}$ is defined over the set $\Gamma_{\vec K}$ of all job types, by integrating the density over the box $I_\vi$, such that 
    \begin{equation}
        \label{equ： discretized arrival measure}
        \eta^A_{\vec K} (\vi) : =  \eta^A(I_{\vi})  = \int_{I_\vi} \lambda f_V(x) dx, \quad \forall \vi \in  \Gamma_{\vec K}.
    \end{equation}
    The \textbf{$\vec {K}$-arrival rate vector} $\Lambda_{\vec K} \in \R^{\Gamma_\K}$ with indices in $\Gamma_\K$, which characterizes $\eta^A_\K$, is then defined such that $\Lambda_{\vec K}(\vi) =\eta^A_{\vec K} (\vi)$.
\end{definition}

\begin{definition}[Service Option Distribution and Discretized Service Measure]
\label{def: service measure}
    For any fixed $\vec {K}$-discretization, we obtain a set of all admissible service options $C_{\vec K}$ \eqref{equ: C_K d-dim}. We let $\beta_\K$ denote a generic service option distribution over $C_\K$. The measure induced by $\beta_\K$ is defined over the set $\Gamma_{\vec K}$, the mass of which is given by 
    \begin{equation}
        \label{equ: K service measure by M_k}
        \eta^S_{\vec K} (\vi)=\E^{\beta_\K} \lbr M^{(\vi)}\rbr= \sum\limits_{M\in C_\K} \beta_\K(M)  M^{(\vi)},\quad \forall i \in  \Gamma_{\vec K}.
    \end{equation}
    We call $\eta_\K^S$ \textbf{the discretized service measure} induced by $\beta_\K$. 
\end{definition}

\subsubsection{The $\K$-discretized MaxWeight Scheduling Policy}
\label{sssec: K-discretized MaxWeight Scheduling Policy}

In the single-resource setting, with a chosen discretization parameter $K\in \N$, the \textbf{$K$-discretized MaxWeight scheduling policy} ($K$-MW) at time $t$ chooses the service option
\begin{equation}
    \label{equ: K-discretized maxweight}
    M^{MW\text{-}C_K}( t ) : =\arg\max\limits_{M\in C_K} \langle M, q( t ) \rangle.
\end{equation}

Similarly in the $d$-resource setting, with some preset $\K\in \N^d$, we define the $\vec {K}$-discretized MaxWeight scheduling policy ($\K$-MW), which at time $t$ chooses the service option
\begin{equation}
    \label{equ: K-discretized maxweight d-dim}
    M^{MW\text{-}C_\K}( t ) :=\arg\max\limits_{M\in C_{\vec K}} \langle M, q( t ) \rangle =\arg\max\limits_{M\in C_{\vec {K}}} \sum\limits_{\vec{i}\in \Gamma_{\vec{K}}} q_{\vec {i}}( t ) M^{(\vec {i})}.
\end{equation}

We need to carefully select $\vec {K}$ by considering the tradeoff between achieving a better mean response time and obtaining a larger stability region. A large $\K$ makes our policy close to naive MaxWeight, which is equivalent to Least-Server-First until very large queue length, and leads to a worse mean response time. However, a small $\K$ causes a coarse discretization that treats every job like the requirement is maximum for the type. The waste of capacity in rounding up results in a poor stability region.


One of our main focuses is to show that the $\K$-MW family of policies is throughput-optimal, which means if there is a policy that can stabilize the system, then there exists some $\K$ such that $\K$-MW can also stabilize the system. \cref{lem: suff and necc cond for stability} will provide a sufficient stability condition for $\K$-MW. \cref{thm: T-O of K-MW given Lipschitz condition} will prove the throughput-optimality of this policy given that the p.d.f. of the resource requirement $V$ is of finite Lipschitz constant.

\subsubsection{$\K$-discretized Efficient MaxWeight}
\label{sssec: Efficient discretized MaxWeight}

One drawback of $\K$-MW is its large candidate set, resulting in poor empirical efficiency when $\K$ is large.
The cardinality of the set $C_\K$ of all admissible service options grows exponentially with the square root of $\max\{K_i\}$. For example, in the single-resource setting, the number of full-utilization service options is $\tilde \theta \lpar \exp(\pi \sqrt \frac{2K}{3})\rpar$ \cite{hardy1918asymptotic}. Therefore, when we choose a large discretization parameter $\K$, a complicated optimization problem must be solved every time the state changes.

Faced with this problem, we propose a family of policies called \textbf{$\K$-discretized Efficient MaxWeight} ($\K$-EMW). Each policy in this family specifies an efficient set of service options $E_{\vec K}\subset C_\K$. The policy only searches over candidates in $E_\K$ and selects the following service option:
\begin{equation}
    \label{equ: K-efficient policy d-dim}
    M^{MW\text{-}E_{\vec K}}( t ) : =\arg\max\limits_{M\in E_{\vec K}} \langle M, q( t ) \rangle.
\end{equation}

We will choose the following four efficient sets that still guarantee throughput-optimality: the 2-Job efficient set $E^{2J}_K$, the 2-Bucket efficient set $E^{2B}_K$, the Extreme-vertices efficient set $E^{X}_K$, and the Pairwise Extreme-vertices (XP) efficient set $E^{\text{XP}}_K$.

\textit{The 2-Job efficient set $ E^{2J}_\K$:}
    We focus on this efficient set in settings with symmetric resource distribution $V$.
    In the single-resource setting, $E^{2J}_K$ consists of all candidates that serve one type-$j$ job and one type-$(K-j)$ job together, and the candidate that serves one type-$K$ job. More specifically, if $K$ is odd, we define $ E^{2J}_K = \{ M_K, M_1,..., M_{\frac{K-1}{2}}\}$, 
    such that $M_K  =\{K\}$, and $M_j =\{j,K-j\}$ for each $j =1,..., \frac{K-1}{2}$. For even $K$, we add the option that serves two type-${\frac{K}{2}}$ jobs.
    
    In the $d$-resource setting and given $\vec K$-discretization, we define $E^{2J}_\K$ for a $\K$ discretization as follows. Recall that $\vec K=(K_1,,,.K_d) \in \N^d$. We consider the candidates that serve one type-$\vi$ job and one type-$(\K-\vi)$ job together, and the candidates that serve only one job that has at least one resource requirement reaching the capacity limit (``boundary" job).
    To be specific, assume $K_1,..., K_d$ to be odd. For any $\vj\in \mathcal B^{\vec K}$, the boundary set, such that
    \begin{equation}
        \label{equ: boundary set}
        \mathcal B^{\vec K} = \{ \vec j : j_\ell = K_\ell \quad \text{for some}\; \ell = 1,..., d\},
    \end{equation}
    we set $M_{\vj}= \{\vj\}$. Then for each $(\frac{K_1+1}{2} ,1,...,1) \preceq \vj \preceq (K_1-1, ..., K_d-1)$, we let $M_{\vj} = \{\vj,\K-\vj\}$. We handle even $\K$ as in the $d=1$ case.

    The MaxWeight policy that uses $E^{2J}_\K$ as the candidate set is then called \textbf{2-Job $\K$-Efficient MaxWeight} (2J $\K$-EMW).

    \textit{The 2-Bucket efficient set $E^{2B}_K$:}
    We focus on $E^{2B}_K$ in $d=1$ single-resource setting, where the requirement density $f_V$ is decreasing.
    We let $K = 2^L$ for some integer $L$. $E^{2B}_K$ contains the following candidates, expressed by the alternative notation II:
    \begin{itemize}
        \item For $k = 2^\ell\in [K]$ that is a power-of-two integer (including $k=1$), we include the service option that serves $2^{L-\ell}$ type-$k$ jobs simultaneously. Specifically, $M_{2^\ell} = \{(2^\ell, 2^{L-\ell})\}$.
        \item For $k \in [K]$ that is not a power-of-two integer, let $\ell:=\ell(k) = \lceil \log_2 j \rceil$. We include the service option that serves $2^{L-\ell}$ type-$k$ jobs and $2^{L-\ell}$ type-$(2^\ell-k)$ jobs simultaneously. Specifically, $M_k = \{(k,2^{L-\ell}),(2^\ell - k, 2^{L-\ell})\}$.
    \end{itemize}
    We call the MaxWeight policy with efficient set $E^{2B}_K$ 2\textbf{-Bucket K-Efficient Discretized MaxWeight} (2B K-EMW).
    
    \textit{The Extreme-vertices efficient set $E^X_\K$:}  $E^X_\K$ is applicable under any requirement distribution $V$.
    For any $\K$ fixed, we let $E^X_\K$ denote the set of all extreme vertices in $C_\K$ (the extreme points in $ConvH(C_\K)$). The corresponding policy that chooses $M^{X} = \arg \max_{M \in E^X_\K} \langle M,q(t)\rangle$ is called \textbf{Extreme-vertices $\K$-Efficient Discretized MaxWeight} (X $\K$-EMW).

    \textit{The Pairwise Extreme-vertices (XP) efficient set $E^{\text{XP}}_K$:}
    $E^{\mathrm{XP}}_K$ is applicable under the MSJ setting when the discretization parameter $K$ is even. It is designed as a tractable surrogate for the X efficient set $E^X_K$, because the exact construction of $E^X_K$ is computationally difficult. Specifically, $E^{\mathrm{XP}}_K$ is a superset of $E^X_K$ and can be constructed by a direct set-subtraction procedure.
    
    Formally speaking, for any $N\leq K$, we define the set $\tilde C_N$, which represents the set of service options with total resource usage $N/K$:
    \[\tilde C_N : = \{M\in \N^{K}: \sum\limits^K_{k=1} k M^{(k)} = N \}.\]
    The XP efficient set is obtained from $\tilde C_K$ by excluding every service option that can be decomposed into two distinct service options, each with total resource usage $1/2$:
    \begin{equation}
        \label{equ: XP efficient set}
        E^{\text{XP}}_K := \{ M \in \tilde C_K: M \neq M_1+M_2, \text{ where } M_1\neq M_2 \text{ and } M_1, M_2 \in \tilde C_{\frac{K}{2}}\}.
    \end{equation}
    We call the MaxWeight policy using $E^{\mathrm{XP}}_K$ \textbf{Pairwise Extreme-vertices $K$-Efficient Discretized MaxWeight} (XP $K$-EMW).
    

\subsubsection{$\K$-Discretized Nonpreemptive Markovian Service Rate}
\label{sssec: K-discretized Markovian Service Rate}

In addition to the above preemptive policies, we provide nonpreemptive policies to schedule jobs with continuous resource requirements. This family of policies is built from the nonpreemptive Markovian Service Rate (nMSR) policy \cite{chen2024analyzing}. MSR  uses an irreducible, finite-state continuous-time Markov chain $\{M( t )\}_{t\geq 0}$ as a modulating process. The state space of $\{M( t )\}_{t\geq 0}$ is a set of service options. As proven in \cite{chen2024analyzing}, the family of the nMSR policies is throughput-optimal. In addition, the only computational expense is the precomputation to find the modulating process -- executing the process is more efficient.

In our setting, where the resource requirement distribution is continuous, we again use $\K$-discretization as introduced in \cref{ssec: Auxiliary Notations}, and consider the set of all schedulable discretized service options $C_\K$. By conducting an offline pre-computation, we choose at most $|\Gamma_\K|$ service options in $C_\K$ to form the state space of the modulating process $\{M( t )\}_{t\geq 0}$. Then nMSR can be generalized to the continuous setting, called \textbf{$\K$-discretized nMSR} ($\K$-nMSR). We will show in \cref{thm: MSR+K-MW} that $\K$-nMSR is throughput-optimal, given that the p.d.f. of the resource requirement is Lipschitz continuous.

\subsubsection{$\K$-Discretized Efficient Nonpreemptive Markovian Service Rate}
\label{sssec:KE nMSR}

To reduce the complexity of the offline computation in selecting the modulating process, the efficient set setting can be extended to the nonpreemptive case. Specifically, after $\K$-discretization, nMSR chooses service options from an efficient set. We recall \cref{sssec: Efficient discretized MaxWeight} and define efficient nMSR policies corresponding to the 2-Job (2J $\K$-EnMSR), 2-Bucket (2B $\K$-EnMSR), Extreme-vertices (X $\K$-EnMSR), and Pairwise Extreme-vertices (XP $\K$-EnMSR) efficient sets.

\subsection{Discretization-Based Policies with Backfilling}
\label{ssec: Backfilling}

To fully utilize the available capacity, Backfilling can be applied for each preemptive discretization-based policy. Backfilling works the same way as the filling rule of First-Fit. Once the underlying policy selects a candidate service option, the remaining unused resources may still allow some jobs in the system to be packed. The server then scans the queue in arrival order and packs jobs by the following greedy rule. The server packs a scanned job and serves it together with the jobs already selected, whenever the resource requirements of the scanned job are at most the remaining resources of the server. For brevity, we refer to $\K$-discretized MaxWeight with Backfilling as $\K$-MW-B. We note that $\K$-MW-B or $\K$-EMW-B are not purely discretization-based policies.

\cref{cor: Backfilling Maintains Stability} states that the same type of proof of stability for $\K$-MW (or $\K$-EMW) also guarantees the stability of $\K$-MW-B (or $\K$-EMW-B). In addition, adopting Backfilling improves performance without significantly worsening computation efficiency, which is verified in \cref{sec: empirical results}.

\section{Main Results}
\label{sec: main results}

Our main goal is to find the first throughput-optimal family of policies in the continuous MRJ setting. We begin by stating our core result: In both single- and multi-resource settings, the family of $\K$-discretized MaxWeight ($\K$-MW) policies is throughput-optimal for all resource distributions $V$ with Lipschitz continuous p.d.f.s $f_V$.
\begin{theorem}[Throughput-optimality of $\K$-MW]
    \label{thm: T-O of K-MW given Lipschitz condition}
    In the continuous MRJ setting with exponential durations, for any resource requirement distribution $V$ with Lipschitz continuous probability density function $f_V$ over $(0,1]^d$, the family of $\K$-MW policies is throughput-optimal. That is, if the continuous MRJ system is stable under some policy $\pi$, then there exists some $\K$ such that $\K$-MW stabilizes the system. 
\end{theorem}

The proof is deferred to \cref{ssec: proof of T-O KMW}. 

\noindent\textit{Proof Sketch:} We show that under the discretized $\K$-MW policy, we can switch our focus to the discretized CTMC $q_t$ introduced in \cref{ssec: Auxiliary Notations}. We build upon the standard stability argument for class-based policies: It suffices to exhibit a distribution $\beta_\K$ over service options that fulfills a discrete measure dominance condition. The condition states that the total service rate for each class under $\beta_\K$ exceeds its arrival rate (\cref{lem: suff and necc cond for stability}). 

The challenge arises because the conversion to the discrete system involves rounding up the resource requirements, increasing the system load. As a result, the system may fail to be stable.
To overcome this, we must both choose $\K$ carefully based on the Lipschitz constant of $f_V$, as well as carefully construct $\beta_\K$.

To construct $\beta_\K$ satisfying discrete measure dominance, we use the fact that there exists a stable continuous policy $\pi$, which induces a continuous measure dominance for the undiscretized system. However, the naive construction of $\beta_\K$ fails. Simply rounding up each continuous service option leads to inadmissible service options. We instead round \emph{down} each continuous service option to generate admissible service options, and invoke Lipschitz continuity of $f_V$, to argue that the density of the rounded-up and rounded-down buckets are similar. Our approach requires special care for certain job types, including
\begin{itemize}
    \item Job types that use the entirety of some resource,
    \item Job types on the boundary of the support of $V$, and 
    \item Low-density job types.
\end{itemize}
By carefully handling these job types, we demonstrate that we can always choose a $\K$ and a $\beta_\K$ such that the discrete measure dominance is satisfied, and conclude the theorem. \hfill $\square$

When jobs are not preemptible, a nonpreemptive family of policies is needed. In \cref{thm: MSR+K-MW}, we show that the family of $\K$-nMSR policies is the first nonpreemptive throughput-optimal family in the continuous MRJ setting, under the same mild assumption of \cref{thm: T-O of K-MW given Lipschitz condition}.

\begin{theorem}[Throughput-optimality of $\K$-nMSR]
    \label{thm: MSR+K-MW}
    In the continuous MRJ setting, the family of $\K$-nMSR policies is throughput-optimal for all $V$ with Lipschitz continuous density $f_V$.
\end{theorem}

The proof is deferred to \cref{ssec: Proof of T-O KnMSR}. The proof follows similar steps to \cref{thm: T-O of K-MW given Lipschitz condition}, because the stability condition of $\K$-nMSR (\cref{lem: Chen 2024 analyzing T-O nMSR}) is equivalent to that of $\K$-MW, as we described in \cref{lem: equivalent suff and neec conds}.

The candidate set $C_\K$ of all schedulable $\K$-discretized service options is extremely large for large $\K$. To execute $\K$-MW, the server has to solve a complicated optimization problem every time the state changes. Similarly, to execute $\K$-nMSR, the server requires to perform an expensive offline pre-computation to select the modulating process.

To avoid the computation expense of searching over all of $C_\K$, we provide more efficient families of policies by considering smaller subsets of candidate service options. These subsets of $C_\K$ are called efficient sets. Recall first the extreme-vertices set $E^X_\K \subset C_\K$ from \cref{sssec: Efficient discretized MaxWeight}, which is the set of all extreme points of $ConvH(C_\K)$. The following theorem states that the corresponding discretized MaxWeight and nMSR policies are throughput-optimal under the same assumption as \cref{thm: T-O of K-MW given Lipschitz condition}. 

\begin{theorem}[Throughput-optimality of X $\K$-EMW and X $\K$-EnMSR]
    \label{thm: vertices KEMW and KEnMSR}
    In the continuous MRJ setting, given the same conditions as in \cref{thm: T-O of K-MW given Lipschitz condition}, X $\K$-EMW and X $\K$-EnMSR are both throughput-optimal. Moreover, for any $\K \in \N^d$ fixed, $\K$-MW, X $\K$-EMW, $\K$-nMSR, and X $\K$-EnMSR all are stable for the same set of arrival rates.
\end{theorem}

The proof of \cref{thm: vertices KEMW and KEnMSR} is deferred to \cref{ssec: Proof of T-O Vertices Version}. The efficiency of X $\K$-EMW and X $\K$-EnMSR lies in the fact that the ratio $|E^X_\K|/|C_\K|$ vanishes when $\K$ grows. For example, in the single-resource setting, when $K=60$, both X efficient policies consider $0.53\%$ as many service options as $60$-MW, and when $K=100$, the ratio drops to $0.031\%$ (\citet{shlyk2023number}). However, while $E^X_\K$ significantly reduces the number of service option candidates, its cardinality still goes superpolynomially in $\max (\K)$, so there is still the room for efficiency.

In the search for yet more efficient policies, we turn to a simpler resource distribution in the MSJ (single-dimensional) setting. We consider resource distributions $V$ such that $f_V$ is weakly decreasing. This allows us to achieve throughput-optimality with an efficient set, whose size is $K$, in contrast to the super-polynomial candidate set sizes of $C_k$ and $E^X_K$. Recall the 2-bucket efficient sets from \cref{sssec: Efficient discretized MaxWeight}, with corresponding policies 2B $K$-EMW and 2B $K$-EnMSR. We show both are throughput-optimal in this decreasing density setting. 

\begin{theorem}[2B $K$-EMW and 2B $K$-EnMSR]
    \label{thm: 2DK-EMW decreasing}
    In the continuous MSJ setting with the weakly decreasing density function of $V$, the 2B $K$-EMW and 2B $K$-EnMSR families are throughput-optimal. In particular, for any $\lambda <\frac{1}{\E V}$, there exists some $K$ such that the queuing system is stable under 2B $K$-EMW and under 2B $K$-EnMSR. In addition, it is enough to let $K=2^L$, where $L = \lfloor -\log_2\lpar \frac{1}{\lambda}-\E V \rpar \rfloor+1$.
\end{theorem}

The proof of Theorem is deferred to \cref{ssec: Proof of Decreasing 2D Version}.

\noindent\textit{Proof Sketch:} We first derive the upper bound of the stability region in \cref{lem: MRJ stability region subset}, by showing no policy is stable when $\lambda\geq \frac{1}{\E V}$. Next, we use \cref{lem: suff and necc cond for stability} to prove that when $\lambda <\frac{1}{\E V}$, there exists some $K$ such that both 2B $K$-EMW and 2B $K$-nMSR are stable. The key idea in our construction of $\beta_\K$ over $E^{2B}_K$ to achieve discrete measure dominance is to greedily assign mass to each service option in a clever order. We assign mass starting with service options serving jobs with higher resource requirements down to lower requirements. The weakly decreasing $f_V$ allows us to construct a $\beta_\K$ distribution over service options, guaranteeing discrete measure dominance while only using service options that fully utilize all resources.
\hfill $\square$ 

\cref{thm: 2DK-EMW decreasing} does not need to assume any continuity condition for $f_V$ Lipschitz or otherwise. Our choice of $K$ is fully explicit, depending only on the arrival rate $\lambda$ and the expectation $\E V$. 

We then focus on another setting, where the resource requirement distribution $V$ is centrally symmetric and is of arbitrary dimension. We have the 2-Job efficient set $E^{2J}_\K$, introduced in \cref{sssec: Efficient discretized MaxWeight}. $E^{2J}_\K$ contains $\prod \K/2$ service options, half as many as the number of job types. Recall from \cref{sssec: Efficient discretized MaxWeight,sssec:KE nMSR}, we call the corresponding efficient MaxWeight policy and nMSR policy 2J $\K$-EMW and 2J $\K$-EnMSR, respectively. \cref{thm: 2JK-EMW Unif} shows the throughput-optimality result for 2J $\K$-EMW and 2J $\K$-EnMSR, for resource requirement $V$ such that its p.d.f. $f_V$ is Lipschitz continuous and centrally symmetric with respect to $\frac{1}{2} \vec 1$.

\begin{theorem}[2J $\K$-EMW and 2J $\K$-EnMSR]
    \label{thm: 2JK-EMW Unif}
    In the continuous MRJ setting, assume that the p.d.f. $f_V$ is Lipschitz continuous and centrally symmetric with respect to $\frac{1}{2} \vec 1$. In this case, both the 2J $\K$-EMW and 2J $\K$-EnMSR families are throughput-optimal. To be specific, for any arrival rate $\lambda <2$, there exists some $\K$ such that both 2J $\K$-EMW and 2J $\K$-EnMSR stabilize the system. 
     
    Moreover, if $V \sim Unif(( 0,1]^d)$, any $\K \succeq K \vec 1$ stabilizes the system, where $K$ is the smallest odd integer such that $K \geq \lfloor \frac{2\lambda d}{2-\lambda} \rfloor+1$. Additionally, in the single-resource setting, it suffices to choose an odd $K$ such that $K\geq \lfloor \frac{\lambda}{2-\lambda}\rfloor +1$.
\end{theorem}

The proof of \cref{thm: 2JK-EMW Unif} is deferred to \cref{ssec: Proof of Unif 2J Version}. 

\noindent\textit{Proof Sketch:} On the one hand, we use \cref{lem: MRJ stability region subset} to show that no policy is stable when $\lambda \geq 2$. On the other hand, for any $\lambda<2$, we construct discrete measure dominance and apply \cref{lem: suff and necc cond for stability}. The key idea is to assign mass to the option serving the type-$\vi$ jobs and the type-$(\K-\vi)$ equal to the larger arrival rate of the two job types. The two arrival rates differ negligibly due to the symmetry Lipschitz condition of $f_V$. We also separately handle the boundary set, involving jobs that fully utilize a resource. By choosing an appropriately large $\K$, which depends both on the Lipschitz constant and the size of the boundary set, we construct a distribution $\beta_\K$ satisfying the discrete measure dominance condition, similar to the proof of \cref{thm: T-O of K-MW given Lipschitz condition}. \hfill $\square$

In the single-dimensional (MSJ) setting, the efficient set has cardinality $|E^{2J}_K| = \frac{K+1}{2}$, which grows linearly in $K$, compared with $|C_K|$'s exponential growth in $K$. In the multi-dimensional (MRJ) setting, when we let $\K = K\vec 1$, the cardinality of the 2-Job $K$-efficient set $|E^{2J}_\K|$ grows exponentially in $d$ and polynomially in $K$, whereas $|C_\K|$ grows exponentially both in $K$ and $d$.

We note in \cref{cor: Backfilling Maintains Stability}, that our stability condition for a MaxWeight-type policy also suffices to demonstrate stability when Backfilling is incorporated, as described in \cref{ssec: Backfilling}. Consequently, under the same conditions as \cref{thm: T-O of K-MW given Lipschitz condition,thm: vertices KEMW and KEnMSR,thm: 2DK-EMW decreasing,thm: 2JK-EMW Unif}, the corresponding families of $\K$-MW-B and $\K$-EMW-B are also throughput optimal. In \cref{sec: empirical results}, we further examine the empirical performance. In addition, we empirically demonstrate that our policies achieve the state-of-the-art performance, both on parametric distribution and datacenter traces. 

\section{Proofs}
\label{sec: proofs}

In this section, we provide the full proofs of our main results in \cref{sec: main results}. 

In \cref{ssec: Stability Condition for MaxWeight Family}, we start by proving a general sufficient condition for stability (\cref{lem: suff and necc cond for stability}) for $\K$-discretized MaxWeight ($\K$-MW) and Efficient $\K$-MW ($\K$-EMW), with given $\K$ and efficient set $E_\K$. We also show in \cref{cor: Backfilling Maintains Stability} that this condition suffices for stability for $\K$-MW with Backfilling ($\K$-MW-B) and $\K$-EMW with Backfilling ($\K$-EMW-B).

In \cref{ssec: proof of T-O KMW}, we prove our first main result, \cref{thm: T-O of K-MW given Lipschitz condition}, demonstrating throughput-optimality for the $\K$-MW family. To do so, we state and prove \cref{lem: Lipschitz condition induces discretized dominance}, which guarantees the existence of a discretization parameter $\K$ based on the arrival measure satisfying the stability condition in \cref{lem: suff and necc cond for stability}. We then use \cref{lem: nec cond of stability,lem: Lipschitz condition induces discretized dominance} to show the throughput-optimality of the $\K$-MW family (\cref{thm: T-O of K-MW given Lipschitz condition}).

In \cref{ssec: Proof of T-O KnMSR}, we connect the sufficient stability condition in \cref{lem: suff and necc cond for stability} with the existing stability condition of $\K$-nonpreemptive Markovian Service Rate ($\K$-nMSR) \cite{chen2024analyzing}. We then combine the stability condition with \cref{lem: Lipschitz condition induces discretized dominance} and show the throughput-optimality of the $\K$-nMSR family (\cref{thm: MSR+K-MW}).

In \cref{ssec: Proof of T-O Vertices Version}, we prove \cref{thm: vertices KEMW and KEnMSR}, which states the throughput-optimality of the Extreme-vertices $\K$-EMW (X $\K$-EMW) family and X $\K$-EnMSR family. We also show the throughput-optimality of the Pairwise Extreme-vertices $\K$-EMW (XP $\K$-EMW) family and XP $\K$-EnMSR family, which is a consequence of our proof. 

In \cref{ssec: Necessary Stability Condition for an MRJ System: superset}, we prove \cref{lem: MRJ stability region subset}, which gives an upper bound of the stability region for any MRJ system.
In \cref{ssec: Proof of Decreasing 2D Version}, we prove the 2-Bucket $K$-EMW (2B $K$-EMW) family and the 2B $K$-EnMSR family, when $V$ has single dimensional decreasing density (\cref{thm: 2DK-EMW decreasing}).
Finally, in \cref{ssec: Proof of Unif 2J Version}, we focus on symmetrical distribution of $V$ and show the throughput-optimality of the 2-Job $\K$-EMW (2J $\K$-EMW) family and the 2J $\K$-EnMSR family (\cref{thm: 2JK-EMW Unif}).
The core idea is that, for any arrival rate below the upper bound given in \cref{lem: MRJ stability region subset}, we explicitly construct some $\K$ and an appropriate distribution $\beta_\K$ over each efficient set, such that the stability conditions are satisfied.

\subsection{Sufficient Stability Condition for $\K$-MW and $\K$-EMW Families}
\label{ssec: Stability Condition for MaxWeight Family}

We now provide our sufficient stability condition for ${\vec K}$-MW and ${\vec K}$-EMW.

\begin{lemma}
\label{lem: suff and necc cond for stability}
    In the continuous MRJ setting, for a fixed $\K \in \N^d$, the $\K$-MW policy is stable if there exists a service option distribution $\beta_\K$ such that its corresponding service measure $\eta_{\vec K}^S$ (see \cref{def: service measure}), satisfies the following discrete measure dominance condition
    \begin{equation}
        \label{equ: condition discrete dominance mu and nu}
        \eta_{\vec K}^S (\vec {i}) > \eta^A_{\vec K}(\vec {i}), \quad \forall \vec {i} \in  \Gamma_{\vec K}.
    \end{equation}
    Moreover, the $\K$-EMW policy with efficient set $E_{\vec K}: = \{M_j\}_{j\in \mathcal J}$ is stable if the condition \eqref{equ: condition discrete dominance mu and nu} holds, and the support of $\beta_\K$ is a subset of $E_{\vec K}$.
\end{lemma}

\begin{proof}
    Given that the measure dominance condition \eqref{equ: condition discrete dominance mu and nu} holds, we will use the Foster-Lyapunov Theorem \cite{meyn1993stability} to prove the stability of $\{q( t )\}_{t\geq 0}$, the CTMC defining the $\K$-discretized system in \cref{ssec: Auxiliary Notations}. In this context, we define the Lyapunov function $L: \N^{\Gamma_\K} \rightarrow \R$ such that
    \begin{equation}
        \label{equ: Lyapunov function discretized case}
        L(q) = \norm{q}_2^2 = \sum_{\vec i\in \Gamma_{\vec K}} q_{\vec i}^2.
    \end{equation}
    We let $\Delta_M \circ L(q)$ denote the drift \eqref{equ: drift of V} of $L(q)$ under our $K$-EMW policy for some efficient set, which also applies to the full $K$-MW. The Foster-Lyapunov Theorem states that the CTMC $\{q( t )\}_{t\geq 0}$ is positive recurrent if there exist constants $c>0$, $b>0 $ and a finite set $C \subset \N^{\Gamma_\K}$ such that
    \begin{equation}
        \label{equ: FL in the CTMC}
        \Delta_M\circ L(q) \leq -c + b \one{\{q \in C\}}, \quad \text{for all }q \in \N^{\Gamma_\K}.
    \end{equation}

    We define the random vector $A(t) := (A_{\vec i}( t )) \in \N^{\Gamma_{\vec K}}$, where for each job type $\vi \in \Gamma_\K$, the entry $A_\vi(t)$ denotes the numbers of type-$\vec i$ jobs arriving in the time period $(0,t]$. We define the potential completion $S(t)\in \N^{\Gamma_{\vec K}}$, and the realized completion $S(t)\in \N^{\Gamma_{\vec K}}$ similarly, such that $\tilde{S}_{\vec i}(t)$ and ${S}_{\vec i}(t)$ denote the number of type-$\vi$ jobs our policy attempts to complete and actually completes, respectively. 
    
    We now clarify the distinction between $\tilde S_\vi(t)$ and $S_\vi(t)$. During a time period with service option $M$ that attempts to serve $M^{(\vi)}$ type-$\vi$ jobs, the potential completion process of job type $\vi$ is a Poisson process $PP(M^{(\vi)})$. However, the system may have less than $M^{(\vi)}$ type-$\vi$ jobs during this period, resulting in the realized completions $\tilde{S}_{\vec i}(t)$ being less than the potential completions $S_\vi(t)$. Hence, we note that $\tilde{S}_{\vec i}(t) \leq S_{\vec i}(t)$ because some potential service may go unused. 
    
    We now upper bound the drift $\Delta_M \circ L$ of our Lyapunov function. Given any $q \in \N^{\Gamma_{\vec K}}$,
    \begin{align}
        & \Delta_M \circ L(q)  = \lim\limits_{t\downarrow 0} \frac{1}{t} \E \lbr V(q( t )) - V(q(0))\,|\, q(0) = q\rbr \nonumber \\
        & = \lim\limits_{t\downarrow 0} \frac{1}{t} \E \Big[ \norm{q(0)+A( t ) - \tilde{S}( t ) }_2^2 - \norm{q(0)}_2^2 \,\bigg |\, q(0) = q \Big] \nonumber\\
        & =\lim\limits_{t\downarrow 0} \frac{1}{t}  {\E \Big[ \norm{A( t ) - \tilde{S}( t )}_2^2  \,\Big |\, q(0) = q\Big]}  + 2\lim\limits_{t\downarrow 0} \frac{1}{t}{ \E  \Big[\langle q(0), A( t ) - \tilde S( t ) \rangle \,\Big |\, q(0) = q \Big]} \nonumber \\
        & =\lim\limits_{t\downarrow 0} \frac{1}{t} \bigg\{ \underbrace{\E \Big[ \norm{A( t ) - \tilde{S}( t )}_2^2  \,\Big |\, q(0) = q \Big]}_{( I )}  + 2\underbrace{ \E  \Big[ \langle q(0), A( t ) -  S( t ) \rangle \,\Big |\, q(0) = q\Big]}_{(II)} \nonumber  \\
        & \quad \quad\quad \quad\quad + 2\underbrace{ \E  \Big[\langle q(0), S( t ) - \tilde S( t ) \rangle \,\Big |\, q(0) = q\Big]}_{(III)} \bigg \} \label{equ: sum of three items}
    \end{align}

We now upper bound $(I)$, $(II)$, and $(III)$. For $( I )$, by letting $K_m := \min\{K_1,...,K_d\}$, we have that
\begin{allowdisplaybreaks}[0]
\begin{align}
    ( I ) & = \E \Big[ \sum\limits_{\vec i \in \Gamma_{\vec K}} \Big( A_{\vec i}( t ) - \tilde S_{\vec i}( t ) \Big)^2 \, \Big| \, q(0) = q\Big] \nonumber\\
    & = \sum_{\vec i \in \Gamma_{\vec K}} \E A_{\vec i}^2( t ) + \E \Big[  \sum_{\vec i \in \Gamma_{\vec K}}  \tilde S_{\vec i}^2( t ) \Big| \, q(0) = q\Big] -2 \sum_{\vec i \in \Gamma_{\vec K}} \E A_{\vec i}( t )\, \E \Big[  \tilde S_{\vec i}( t ) \Big| \, q(0) = q\Big] \nonumber \\
    & \leq \E  \Big[  \Big( \sum_{\vec i \in \Gamma_{\vec K}} A_{\vec i}( t ) \Big)^2 \Big] + \E \Big[ \Big(\sum_{\vec i \in \Gamma_{\vec K}} \tilde S_{\vec i}( t ) \Big)^2 \Big| \, q(0) = q \Big] \label{equ: nonneg A and S tilde} \\
    & \leq \E  \Big[ \Big( \sum_{\vec i \in \Gamma_{\vec K}} A_{\vec i}( t ) \Big)^2 \Big] + \E \Big[\Big( \sum_{\vec i \in \Gamma_{\vec K}}  S_{\vec i}( t ) \Big)^2 \Big| \, q(0) = q \Big] \label{equ: tilde S_k leq S_k} \\
    & \leq \lambda t + \lambda^2 t^2 + K_mt + K_m^2t^2.\label{equ: two PP}
\end{align}
\end{allowdisplaybreaks}
where 
\begin{itemize}
    \item \eqref{equ: nonneg A and S tilde} is due to $A_{\vec i}( t )$ and $\tilde S_{\vec i}( t )$ are nonnegative for all $\vec i$,
    \item \eqref{equ: tilde S_k leq S_k} is due to $\tilde S( t ) \preceq S( t )$, and 
    \item \eqref{equ: two PP} is due to $\sum_{\vec i \in \Gamma_{\vec K}} A_{\vec i}( t )\sim \text{Poisson}(\lambda t)$ and $\sum_{\vec i \in \Gamma_{\vec K}}  S_{\vec i}( t )$ is stochastically dominated by $\text{Poisson}(K_m t)$, because the server can serve at most $K_m$ jobs simultaneously, which in turn holds because the minimum resource $m$ requirement is $1/K_m$.
\end{itemize}
Hence, $ \lim\limits_{t\downarrow 0} \frac{( I )}{t} \leq \lambda +K_m$. Note that this bound holds for any $\K$-discretized policy.

For $(II)$, we primarily focus on $\vec K$-EMW policy \eqref{equ: K-discretized maxweight d-dim} with some efficient set $E_\K$. $\K$-MW can be treated as a special case where the efficient set $E_\K = C_\K$.

The corresponding service process is approximated by an entry-wise independent Poisson distribution over short time intervals, such that $S^{EMW}(t) \sim \text{Poisson}(M^{MW\text{-}E_\K}(0) t)+\epsilon_t$, where $\epsilon_t$ satisfies $\lim\limits_{t\downarrow 0} \pro(\norm{\frac{\epsilon_t}{t}}_2 > \delta) = 0$ for any $\delta>0$. In addition, the discrete measure dominance condition \eqref{equ: condition discrete dominance mu and nu} means that there exist some positive $\delta$ such that for all $\vi \in \Gamma_\K$,
\begin{equation}
    \label{equ: (1+delta) upper bound}
    \eta^S_\K(\vi) = \sum\limits_{j\in \mathcal J} \beta_j  M_j^{(\vec i)} \geq (1+\delta) \eta^A_\K = (1+\delta) \Lambda_\K (\vec i).
\end{equation} 
Therefore, term (II) has the following upper bound:
\begin{allowdisplaybreaks}[0]
\begin{align}
    (II) & = \langle q \,, \,\Lambda_{\vec K} t - M^{MW\text{-}E_\K}(0)t + o( t )\rangle \nonumber \\
    & = \langle q \,, \, \Lambda_{\vec K} t+ o( t ) \rangle -\langle q\,,\, t \sum\limits_{j=1}^J \beta_j  M^{MW\text{-}E_\K}(0) \rangle \nonumber\\
    & = \langle q \,, \, \Lambda_{\vec K} t+ o( t ) \rangle - t \sum\limits_{j=1}^J \beta_j \langle q\,,\,  M^{MW\text{-}E_\K}(0) \rangle \nonumber\\
    &\leq \langle q \,, \, \Lambda_{\vec K} t+ o( t ) \rangle -t \sum\limits_{j=1}^J \beta_j  \langle q\,,\,  M_j \rangle \label{equ: due to def of MW}\\
    & \leq \langle q \,, \, \Lambda_{\vec K} t+ o( t ) \rangle -\langle q\,,\, (1+\delta)\Lambda_{\vec K} t\rangle  \label{equ: use (1+delta) upper bound}\\
    & =  \langle q \,, \, -\delta \Lambda_{\vec K} t+ o( t ) \rangle,\nonumber
\end{align}
\end{allowdisplaybreaks}
where inequality \eqref{equ: due to def of MW} is due to the definition of $\vec K$-EMW, and inequality \eqref{equ: use (1+delta) upper bound} is due to \eqref{equ: (1+delta) upper bound}.
As a result, \[\lim\limits_{t\downarrow 0} (II)/t \leq -\delta \langle q\,,\, \Lambda_{\vec K} \rangle \leq \delta \lambda_{min} \norm{q}_1,\] where $\lambda_{min} = \min\{\pro(V\in I_{\vec i}): \pro(V\in I_{\vec i}) >0, {\vec i} \in \Gamma_{\vec K} \}$ is the smallest positive arrival rate for any job type $\vi$. Here we disregard the buckets $I_\vi$ with $\pro(V\in I_\vi) = 0$ because it is almost never that $q(\vec i) > 0$ occurs for such buckets.

Finally, for item $(III)$ we again consider $\K$-MW as a special case of $\K$-EMW. We have $(III) = \sum_{{\vec i \in \Gamma_{\vec K}}} \E[  q_{\vec i}(S_{\vec i}( t )-\tilde S_{\vec i}( t ) | q(0) = q]$. For each job type $\vi$, notice that $\lim\limits_{t\downarrow 0} \frac{\pro(S_{\vec i}( t )-\tilde S_{\vec i}( t ) > 0)}{t}>0$ only if the queue length for type-$\vi$ jobs is smaller than the number of type-$\vi$ jobs that $M^{MW\text{-}E_\K}$ attempts to serve, i.e., $q_{\vec i}< M^{MW\text{-}E_\K}_{\vec i}(0)$. Therefore, the upper bound for $\lim\limits_{t\downarrow0}(III)/t$ is as follows:
\begin{align*}
    \lim\limits_{t\downarrow0}\frac{1}{t} (III) & \leq \lim\limits_{t\downarrow0} \frac{1}{t}\E  \lbr \langle q(0), S( t ) \rangle \,\Big |\, q(0) = q,\, q \preceq M^{MW\text{-}E_\K} (0)\rbr\\
    & \leq \lim\limits_{t\downarrow0} \frac{1}{t} \sum\limits_{{\vec i \in \Gamma_{\vec K}}} \E  \lbr M^{MW\text{-}E_\K}_{\vec i}(0) \cdot S_{\vec i}( t ) \,\Big |\, q(0) = q, \, q \preceq M^{MW\text{-}E_\K} (0) \rbr\\
    & \leq \lim\limits_{t\downarrow0} \frac{1}{t} \sum\limits_{{\vec i \in \Gamma_{\vec K}}}  (M^{MW\text{-}E_\K}_{\vec i}(0))^2 t   \\
    & = \norm{M^{MW\text{-}E_\K}(0)}_2^2\\
    & \leq K_m^2,
\end{align*}
where the last inequality is due to $\norm{M^{MW\text{-}E_\K}(0) }_2^2 \leq \norm{M^{MW\text{-}E_\K}(0) }_1^2\leq K_m^2$. 

By adding the above three limits for $(I)$, $(II)$, and $(III)$ in \eqref{equ: sum of three items}, we have shown that $\Delta_M \circ L(q) \leq K_m+2 K_m^2+\lambda - 2\delta \lambda_{min} \norm{q}_1$. Hence, for any $c>0$, when 
\begin{equation}
    \norm{q}_1 \geq \frac{K_m+2K_m^2+\lambda+c}{2\lambda_{min} \delta}, \label{equ: outside the L1 norm bound}
\end{equation} 
we can guarantee the drift $\Delta L(q) < -c$. Therefore, we can define the finite set $C$ to contain all $q$ for which \eqref{equ: outside the L1 norm bound} does not hold such that the Foster-Lyapunov condition is satisfied, and the CMTC $\{q(t)\}$ is positive recurrent. With the positive recurrence of the CTMC, we have proven that the empty state of the discretized system is positive recurrent. As a result, the empty state of the original continuous MRJ system is also positive recurrent. This concludes our stability proof.
\end{proof}

We will use \cref{lem: suff and necc cond for stability} in our proofs of \cref{thm: T-O of K-MW given Lipschitz condition,thm: vertices KEMW and KEnMSR,thm: 2DK-EMW decreasing,thm: 2JK-EMW Unif} for the $\K$-MW, X $K$-EMW, 2B $K$-EMW, and 2J $\K$-EMW families, respectively. With the stability condition in \cref{lem: suff and necc cond for stability} proven, we now give some extensions and related comments. First, we provide an equivalent sufficient stability condition for the discrete measure dominance condition \eqref{equ: condition discrete dominance mu and nu}.

\begin{lemma}[Convex Hull Condition]
\label{lem: equivalent suff and neec conds}
    In the continuous MRJ setting, for a fixed $\K \in \N^d$, $\K$-MW  is stable if the $\K$-arrival rate vector $\Lambda_\K$ is in the interior of the convex hull of the schedulable set $C_\K$. That is, the following convex hull condition holds
    \begin{equation}
        \label{equ: condition Lambda in convex hull of C_K KMW}
        \Lambda_{\vec K} \subseteq int\big( ConvH(C_{\vec K})\big).
    \end{equation}
    Likewise, $\K$-EMW is stable if the following convex hull condition holds for the efficient set $E_{\vec K}$
    \begin{equation}
        \label{equ: condition Lambda in convex hull of C_K}
        \Lambda_{\vec K} \subseteq int\big( ConvH(C_{\vec K})\big).
    \end{equation}
    Specifically, the discrete measure dominance condition holds \eqref{equ: condition discrete dominance mu and nu} for some service option distribution $\beta_\K$ over an efficient set $E_\K$ if and only if the convex hull condition \eqref{equ: condition Lambda in convex hull of C_K} holds. Similarly, \eqref{equ: condition discrete dominance mu and nu} holds for some $\beta_\K$ over $C_\K$ if and only if \eqref{equ: condition Lambda in convex hull of C_K KMW} holds.
\end{lemma}

\begin{proof}
    This lemma is a direct consequence of the definition of the convex hull.    
\end{proof}

\begin{remark}[A Larger Efficient Set Maintains Stability]
    \label{rmk: E1 in E2 means E2MW also stable}
    When $\vec K$ is fixed, for any two efficient sets $E_1\subseteq E_2 \subseteq C_\K$, if $\Lambda_\K \subseteq int(ConvH(E_1))$, then clearly $\Lambda_\K \subseteq int(ConvH(E_2))$. Therefore, when the condition \eqref{equ: condition discrete dominance mu and nu} holds for $E_1$, both $\K$-EMW with efficient set $E_1$ and with efficient set $E_2$ stabilize the system.
\end{remark}

In addition to the $\K$-MW and $\K$-EMW policies, we also study $\K$-MW-B and $\K$-EMW-B, which add backfilling to use more of the capacity than the baseline policies. Empirically, $\K$-MW-B and $\K$-EMW-B achieve lower mean response time and a larger empirical stability region. We now prove that $\K$-MW-B and $\K$-EMW-B are stable under the same stability condition in \cref{lem: suff and necc cond for stability}, demonstrating that this empirical behavior is not a happenstance.

\begin{lemma}[Backfilling Maintains Stability]
\label{cor: Backfilling Maintains Stability}
    Under the same assumption of \cref{lem: suff and necc cond for stability}, if the discrete measure dominance \eqref{equ: condition discrete dominance mu and nu} is satisfied by some service option distribution $\beta_\K$ over an efficient set $E_\K$, then the corresponding $\K$-EMW-B is also stable. Likewise, when \eqref{equ: condition discrete dominance mu and nu} is satisfied by $\beta_\K$ over $C_\K$, then $\K$-MW-B is also stable. 
\end{lemma}
\begin{proof}
    In this proof, we focus on $\K$-EMW-B, and the stability for $\K$-MW-B follows from the case when the efficient set $E_\K = C_\K$.
    
    When $\K$ is fixed, we define $\hat S(t)$ such that $\hat S_\vi (t)$ denotes the number of type-$\vi$ jobs actually served in the time period $(0,t]$. We decompose $\hat S_\vi (t)$ such that $\hat S_\vi (t) = \tilde S_\vi (t) + B_\vi(t)$, where $B_\vi (t)$ denotes the additional number of type-$\vi$ jobs served due to Backfilling, and $\tilde S_\vi$ is as defined in \cref{lem: suff and necc cond for stability} for $\K$-EMW. Note that $\K$-EMW-B is not a purely discretization-based policy, and we need to consider the original $\mathcal{S}$-valued Markov process $\{x(t)\}_{t\geq 0}$ to determine which jobs are served by Backfilling. By explicitly letting $q: \mathcal{S} \rightarrow \N^{\Gamma_\K}$ such that $q_{\vi}(x) = \sum_{\ell=1}^{\norm{x}_0} \one{\{x_\ell \in I_\vi\}}$, the Lyapunov function can still be chosen as before, but as a function of $x\in \mathcal S$. To be specific, we recall the Lyapunov function $L$ of the countable state space $\N^{\Gamma_\K}$ \eqref{equ: Lyapunov function discretized case} in the proof of \cref{lem: suff and necc cond for stability}, and likewise denote $\tilde L$ to be the Lyapunov function for the continuous state space $\mathcal S$, such that
    \[ \tilde L(x):= L(q(x)) = \norm{q(x)}_2^2. \]
    
    Then, when the current state is $x$ and thus $q(x) \in \N^{\Gamma_\K}$, the drift $\Delta_B \circ \tilde L(x)$ under $\K$-EMW-B is given by  
    \begin{allowdisplaybreaks}[0]
    \begin{align}
        \Delta_B \circ \tilde L(x) & : = \lim\limits_{t\downarrow 0} \frac{1}{t} \E \Big[ \tilde L\Big(q\big( x(t) \big) \Big) - \tilde L\Big(q\big( x(0)\big) \Big)\,\Big |\, x(0) = x\Big] \nonumber \\
        & = \lim\limits_{t\downarrow 0} \frac{1}{t} \E \lbr \norm{q\big( x(0) \big)+A( t ) - \hat{S}( t ) }_2^2 - \norm{q\big( x(0) \big)}_2^2 \,\Big |\, x(0) = x\rbr \nonumber \\
        & = \lim\limits_{t\downarrow 0} \frac{1}{t} \E \lbr \norm{q\big( x(0)\big)+A( t ) - \tilde{S}( t ) -B(t) }_2^2 - \norm{q\big( x(0)\big)}_2^2 \,\Big |\, x(0) = x\rbr. \nonumber
    \end{align}
    \end{allowdisplaybreaks}
    
    Because $q\big( x(0)\big)+A(t) - \hat{S}(t) \succeq \vec 0$ and $B(t) \succeq \vec 0$ almost surely, we have that the following conditional expectation bound
    \[\E \lbr \norm{q\big( x(0)\big)+A( t ) - \hat{S}( t) }_2^2  \,\Big |\, x(0) = x \rbr \leq \E \lbr \norm{q(0)+A( t ) - \tilde{S}( t ) }_2^2  \,\Big |\, q(0) = q(x)=:q \rbr. \]
    Therefore, at any state $x\in \mathcal S$, the drift $\Delta_B \circ \tilde L(x)$ of $\tilde L$ under $\K$-EMW-B is always upper bounded by the drift $\Delta_M \circ L(q(x))$ of $L(q(x))$ under $\K$-EMW.
    
    With the Lyapunov argument in the proof of \cref{lem: suff and necc cond for stability}, we conclude that the discrete measure dominance condition \eqref{equ: condition discrete dominance mu and nu} is sufficient to apply the Foster-Lyapunov theorem \cite{meyn1993survey}, which guarantees Harris recurrence in this setting. To be specific, the Foster-Lyapunov condition is satisfied for any $c>0$ and a corresponding compact set
    \[C:= \left\{s \in \mathcal S: \|x\|_0 := \|q(x)\|_1 < \frac{K_m+2K_m^2+\lambda+c}{2\lambda_{min} \delta} \right\}.\]
    This argument shows the Harris recurrence of the set $C$. Note that all states in $C$ have a bounded number of jobs, and the completion rate is always at least $1$. As a result, all states in $C$ can move to the empty state in bounded time with uniformly lower-bounded probability, demonstrating positive recurrence to the empty state.
\end{proof}

In addition, by combining \cref{cor: Backfilling Maintains Stability} with the proofs of \cref{thm: T-O of K-MW given Lipschitz condition,thm: vertices KEMW and KEnMSR,thm: 2DK-EMW decreasing,thm: 2JK-EMW Unif} respectively, we can demonstrate the throughput-optimality of the $\K$-MW-B, X $\K$-EMW-B, 2B $K$-EMW-B, and 2J $\K$-EMW-B families, under their respective assumptions. 

\begin{remark}
    \label{rmk: Any backfilling works}
    In this paper, we primarily focus on the Backfilling technique that searches according to the arrival order, which is the Backfilling rule First-Fit uses. However, \cref{cor: Backfilling Maintains Stability} holds for any Backfilling technique. For instance, in the MSJ setting ($d=1$), the server can pack jobs by scanning the queue according to increasing or decreasing order of resource requirement, which are the Backfilling rule for Least-Server-First (LSF) or Best-Fit, respectively.
\end{remark}

\subsection{Proof of \cref{thm: T-O of K-MW given Lipschitz condition}: Throughput-Optimality of the $\K$-MW Family}
\label{ssec: proof of T-O KMW}

The following lemma states that for any resource requirements random variable $V$ with density $f_V$ Lipschitz continuous, the discrete measure dominance condition \eqref{equ: condition discrete dominance mu and nu} can be established with some appropriate discretization parameter $\K$ and a distribution $\beta_\K$ over $C_\K$. Our main result, \cref{thm: T-O of K-MW given Lipschitz condition}, will then follow directly from combining \cref{lem: suff and necc cond for stability} and \cref{lem: Lipschitz condition induces discretized dominance}.

\begin{lemma}
    \label{lem: Lipschitz condition induces discretized dominance}
    In the continuous MRJ setting, assume that the job resource requirements random variable $V$ has a Lipschitz continuous density $f_V$, and assume that given some $\lambda >0$, there exists a policy $\pi$ that stabilizes the queueing system. Then there exists some $\K\in \N^d$ and a distribution $\beta_\K$ over $C_\K$, such that the discretized service measure $\eta^S_\K$ satisfies the measure dominance \eqref{equ: condition discrete dominance mu and nu}.
\end{lemma}

\begin{proof}
    We give the proof in the multi-resource (MRJ) setting, and the single-resource setting (MSJ) can be shown by setting $d=1$.
    
    Suppose that for a given $\lambda >0$ and $V$, there exists some policy $\pi$ (see \cref{def: Markovian policy}) that leads to a stable system. At the steady-state of this Markov process under $\pi$, we obtain a stationary measure $\beta^\pi$ over $\pi(\mathcal S)$, the set of service options used by $\pi$. The stability of $\pi$ implies that $\eta^A = \eta^\pi$, where 
    \[ \eta^\pi(dx) = 
    \int_{\mathscr S} \sum_{r\in\pi(s)} \one{\{r \in dx\}} \; d\beta^\pi(s),
    \] so that the rate of job arrivals with resource requirements in $I$ equals the completion rate of jobs served with resource requirements in $I$, where $I$ is any subset of $( 0,1]^d$. 

    In the steady state, because the empty state is positively recurrent, the Markov process assigns positive probability mass $\epsilon_0 >0$ to the empty job state $\emptyset \in \mathcal S$. At the empty job state, the only service option is $\emptyset \in \mathscr S := \pi(\mathcal S)$. We consider the conditional measure over $(0,1]^d$
    \[\eta^S = \mathcal L (\eta^\pi| \pi \neq \emptyset) = \frac{1}{1-\epsilon_0} \eta^\pi,\]
    and call $\eta^S$ the \emph{service measure}. The service measure is induced by $\{M_\alpha^\pi\}_{\alpha \in \A} = \pi(\mathcal{S})\setminus \emptyset$, namely the family of continuous service options $\pi$ uses. Note that the index set of these service options $\A$ is uncountable. With a probability measure $\mathcal L (\pi| \eta \neq \emptyset)$ over $\pi (\mathcal S)$, we can obtain a probability measure $\zeta$ over $\A$, which we call the reference distribution, and express the service measure as 
    \begin{equation}
        \label{equ: service measure mu S}
        \eta^S(dx) = \int_\A \sum_{r \in M_\alpha^\pi} \one{\{r \in dx\}}\; d\zeta(\alpha).
    \end{equation}
    The above construction guarantees that $\eta^A$ is absolutely continuous with respect to $\eta^S$, and that
    \begin{equation}
        \label{equ: d mu S/ d mu A}
        \frac{d\eta^S}{d\eta^A}(x) \geq 1+\epsilon, \quad \forall x \in int\lpar supp(\eta^A)\rpar,
    \end{equation}
    where $\epsilon = \frac{\epsilon_0}{1-\epsilon_0}$. We call this property ``continuous measure dominance". 

    Now our goal is to find some $\K$ such that $\K$-MW stabilizes the system as well, by applying \eqref{equ: condition discrete dominance mu and nu}, with a service measure $\beta_\K$ over $C_\K$ chosen based on $\beta^\pi$. To be specific, we let $\K = (K,...,K)$ for some $K$ to be determined later, and aim to build a $\K$-discretized service measure $\eta^S_{\K}$ based on some distribution $\beta_\K$ over $C_\K$, such that there exists some $\delta>0$ satisfying
    \begin{equation}
        \label{equ: nu^S_K / nu^A_K in D dim}
        \frac{\eta^S_{\K}(\vi) }{\eta^A_{\K}(\vi)} \geq 1+\delta, \quad \text{ for all } \vi \in \Gamma_{\K} \text{, for which } \eta^A_{\K}(\vi) > 0,
    \end{equation}
    where recall that for each job type $\vi$, $\eta^A_{\K}(\vi) : = \eta^A(I_{\vi})$.

    In our construction of $\beta_\K$, we start by classifying each discretized job type $\vi\in \Gamma_\K$ in either the default set $\mathcal D^\K$ or the exceptional set $\mathcal E^\K$. For any default index $\vi \in \mathcal D^\K$, we rely on the continuous service options $\{M_\alpha^\pi\}_{\alpha\in \A}$ to propose corresponding $\K$-discretized service options to serve type-$\vi$ jobs, which we then modify. We will discuss the modification after we have defined the exceptional set. For any exceptional index $\vi\in \mathcal E^\K$, we include the service option that serves a single type-$\vi$ job. Specifically, we assign a total probability mass $\xi_\K$ of $\beta^\K$ to each service option $\{\vi\}$ for $\vi \in \mathcal E^\K$, such that
    \begin{equation}
    \label{equ: mass on SO serving a job with type in EK}
        \beta_{\K}(\{\vi\}) = (1+\delta) \eta^A_\K(\{\vi\}),\quad \text{for all } \vi \in \mathcal E^\K.
    \end{equation}
    We will then assign the remaining mass to the discrete service option family $\{M\}$ induced by the continuous service option family $\{M^\pi_\alpha\}$, where the probability will be given by \eqref{equ: beta_K on efficient service options} to control the default set $\mathcal D^\K$. We will show that as $K$ becomes large, the arrival measure on $\mathcal E^\K$ vanishes, and for any job type in the set $\mathcal D^\K$, the gap between our service measure $\eta^S_\K(\vi)$ and $\eta^\pi(I_\vi)$ vanishes. This then allows us to prove discrete measure dominance for large enough $\K$. 
    
    We now introduce the exception set $\mathcal E^\K$, which is the union of three different subsets, namely the boundary set $\mathcal B^\K$, the empty-round-up set $\mathcal I^\K$, and the low-probability set $\Jcal{\K}{}(c_1,c_2)$, with parameters $c_1>0$ and $0<c_2<1$. Then we guarantee discrete measure dominance for all $\vi\in \mathcal E^\K$, by including in $\beta_\K$ the corresponding service option $\{\vi\}$ and with mass $\beta_\K(\{\vi\})$ given in \eqref{equ: mass on SO serving a job with type in EK}. To bound the probability assigned to the service options $\{\vi\}$, $\vi \in \mathcal E^\K$, we need to upper bound the total arrival rate of each subset of $\mathcal E^\K$, and further show that the upper bound converges to zero. Specifically, we will show that $\sum_{\vi \in \mathcal E^\K} \beta_\K(\{\vi\}) = (1+\delta)\eta^A_\K(\mathcal E^\K) = O(K^{-c_2})$. We note that $\delta$ can be arbitrarily chosen in $(0,\epsilon)$, which does not affect the existence of $\K$.

    The first subset of $\mathcal E^\K$ is the boundary set $\mathcal B^\K$, which includes all boundary job types defined in \eqref{equ: boundary set}, namely any job type that uses the entirety of any of the resources. Using the Lipschitz condition, we provide in \cref{lem: Upper bound for the Supremum of Lipschitz Density Function} an upper bound on the maximum density $\norm{f}_\infty$, such that $\norm{f}_\infty \leq L:= L(C,d) < \infty$. The value of $L(C,d)$ is given by \eqref{equ: choice of M}, depending only on the Lipschitz constant $C$ and the dimension $d$. We defer the proof to Appendix~\ref{appsec: Upper bound for the Supremum of Lipschitz Density Function}. Therefore, we have that $\eta^A_\K(\mathcal B^\K)\leq \frac{L}{K^d}dK^{d-1} \leq Ld K^{-1}$, because there are at most $dK^{d-1}$ job types in $\mathcal B^\K$, and each type has arrival rate upper bounded by $\frac{L}{K^d}.$ This completes our bound for $\mathcal B^\K$.

    The next subset of $\mathcal E^\K$ is defined as the empty-round-up set $\mathcal I^\K$, which contains all non-boundary job types $\vi$ with $\eta^A_\K(I_\vi)>0$ and $\eta^A_\K(I_{\vi+\vec 1})=0$. To bound the arrival measure on $\mathcal I^\K$, we consider the zero set $\mathcal Z^\K := \{ \vi \in \Gamma_{\K}: \eta^A_\K(\vi) = 0\}$ of all job types, the corresponding buckets of which have zero arrival rate. Then the following inclusion holds:
    \[
    \mathcal I^\K\subseteq \bigcup\limits_{\vi \in \mathcal Z^\K} \{\vi -\vec 1\}.
    \]
    For any $\vi \in \mathcal I^\K$, $f_V(\frac{\vi+\vec 1}{K}) = 0$. Because $f_V$ is Lipschitz with constant $C$, $\sup_{v \in I_\vi} f(v) \leq C\sqrt{d}/K$ for all $\vi \in \mathcal Z^\K $. Therefore, $\eta^A_{\K} (\{\vi - \vec 1 \}) \leq C\sqrt{d} K^{-D-1}$ for any $\vi \in \mathcal Z^\K$, and consequently, \[\eta^A_{\K} (\mathcal I^\K) \leq C\sqrt{d} K^{-D-1} K^D = C\sqrt{d} K^{-1},\]
    where we use the total number of job types to bound the number of job types in $\mathcal I^\K$, i.e., $|\mathcal I^\K| \leq |\Zcal^\K| \leq |\Gamma_\K| = K^D$. This completes our bound for $\mathcal I^\K$.

    The final subset of $\mathcal E^\K$ is the low-probability set $\mathcal{J}^\K(c_1, {c_2})$ for some constants $c_1 >0$, $0<{c_2}<1$, such that 
    \[\mathcal{J}^\K (c_1, {c_2}) = \{ \vi \in \Gamma_\K: \eta^A_{\K}(\vi)\leq c_1 K^{-(D+{c_2})} \}. \]
    The choice of the pair $(c_1,c_2)$ does not affect the existence of $\K$, but affects the final value of $\K$. In this paper, we do not attempt to optimize $\K$. If concrete values of $c_1, c_2$ are needed, $c_1=1$ and $c_2 = 1/2$ suffice. The arrival measure on the set $\mathcal{J}^{\K}$ is then upper bounded by
    \[ c_1 K^{-{c_2-D}} \,|\Jcal{\K}{} (c_1,c_2)|  \leq  c_1 K^{-{c_2}-D} |\Gamma_\K| =c_1 K^{-{c_2}}. \]
    This completes our bound for $\mathcal{J}^\K(c_1, {c_2})$.
    
    As mentioned before, for each exceptional job type $\vi \in \mathcal{E}^\K := \mathcal I^\K \cup \mathcal B^\K\cup \mathcal{J}^\K(c_1, {c_2})$, we include in $\beta_\K$ the service option $\{\vi\}$, which only serves one type-$\vi$ job, and let $\beta_\K(\{\vi\}) = (1+\delta)\eta^A_\K(\vi)$. This guarantees the discrete measure dominance holds for all $\vi \in\mathcal{E}^\K$, i.e., $\eta^S_\K(\vi)\geq (1+\delta)\eta^A_\K(\vi)$. We also let $\xi_\K := \sum_{\vi \in \mathcal E^\K} \beta_\K(\{\vi\})= (1+\delta) \eta^S_\K( \mathcal{E}^\K)$, which denotes the the summed probability mass $\beta_\K$ assigns to these options. The above discussion implies that 
    \begin{align*}
        \xi_\K & = (1+\delta) \eta^A_\K\lpar \mathcal I^\K \cup \mathcal B^\K\cup \mathcal{J}^\K(c_1, {c_2})\rpar \\
        & \leq (1+\delta)  \lbr Ld K^{-1} + C\sqrt{d} K^{-1} + c_1 K^{-c_2} \rbr \\
        & \leq (1+\delta) (Ld+C\sqrt{d}+c_1) K^{-{c_2}},
    \end{align*}
    and thus $\lim\limits_{K\uparrow \infty} \xi_\K = 0$.

    We now switch our focus to constructing $\beta_\K$ to achieve the discrete measure dominance for the default job types $\vi \in \mathcal{D}^\K := \Gamma_\K \setminus \mathcal{E}^\K$. We consider the family of service options $\{M_\alpha^\pi\}_{\alpha\in \A}$ used by the continuous policy $\pi$. We will take advantage of $\{M_\alpha^\pi\}_{\alpha\in \A}$ to construct the discretized service options that we will use for $\beta_\K$. 

    We create a mapping from $M_\alpha^\pi$ to the discretized service options $M_\alpha^\K$. This mapping entry-wise rounds down each resource requirement $r$ to the nearest smaller discretized bucket (job type). Formally speaking, given any $\alpha \in \A$, $M_\alpha^\pi$ contributes a density $\zeta^{\K}(\alpha) = (1-\xi_{\K})\zeta(\alpha)$ to $M^{\K}_\alpha : = \left\{ \lfloor K r\rfloor\right\}_{r \in M_\alpha^\pi}$. By using smaller job types, any valid $M_\alpha^\pi$ is mapped to a valid  $M_\alpha^\K$ that does not exceed any resource capacity constraint. Then for any service option $M\in \{M_\alpha^\K\}_{\alpha\in \mathcal A}$, 
    \begin{equation}
    \label{equ: beta_K on efficient service options}
        \beta_\K(M) = \int_\mathcal A \one{\left\{ M = \left\{ \lfloor K r\rfloor \right\}_{r \in M_\alpha^\pi}\right \}} d \zeta^\K(\alpha),
    \end{equation}
    and we obtain a valid probability distribution $\beta_\K$, where the remaining $\mathcal E^\K$ indices were defined in \eqref{equ: mass on SO serving a job with type in EK}.
    
    The above construction of $\beta_\K$ means that, for each $\vi \in \mathcal D^\K$, $\eta^S_\K(\vi)$ is given by
    \[\eta^S_\K(\vi)  = \sum\limits_{M\in C_\K} M^{(\vi)} \beta_\K(M) = \int_\A \sum\limits_{l=1}^{\norm{M_\alpha^\pi}_0} \one{ \{ M_\alpha^\pi (l) \in I_{\vi+\vec 1}\} } d\zeta^\K(\alpha),
    \]
    which can be lower bounded, due to \eqref{equ: d mu S/ d mu A} and the Lipschitz assumption $\norm{f}_{Lip} = C$, as follows:
    \begin{allowdisplaybreaks}[0]
     \begin{align}
        \eta^S_\K(\vi) &  = \int_\A \sum\limits_{l=1}^{\norm{M_\alpha^\pi}_0} \one{ \{ M_\alpha^\pi (l) \in I_{\vi+\vec 1}\} } (1-\xi_\K)d \zeta(\alpha)\nonumber\\
        & = (1-\xi_\K) \eta^S(I_{\vi+\vec 1})\nonumber\\
        &\geq (1+\epsilon)(1-\xi_\K) \eta^A(I_{\vi+\vec 1})  \nonumber\\
        & =  (1+\epsilon)(1-\xi_\K) \int_{I_{\vi+\vec 1}} \lambda f(x) dx\nonumber \\    
        & \geq  (1+\epsilon)(1-\xi_\K) \int_{I_{\vi+\vec 1}} \lambda \lpar f\lpar x-\frac{\vec 1}{K}\rpar -\frac{C\sqrt{d}}{K} \rpar dx\nonumber\\
        & = (1+\epsilon)(1-\xi_\K) \lbr \eta^A(\vi) - \frac{\lambda C \sqrt{d} }{K^{d+1}}\rbr\nonumber.
    \end{align} 
    \end{allowdisplaybreaks}
    To conclude that $\eta^S_\K(\vi) \geq (1+\delta) \eta^A_\K(\vi)$, we need to show for any $\vi \in \mathcal D^\K$,
    \begin{allowdisplaybreaks}[0]
    \begin{align*}
        & (1+\epsilon)(1-\xi_\K) \lbr \eta^A_\K(\vi) - \frac{\lambda C\sqrt{d}}{K^{d+1}}\rbr \geq (1+\delta) \eta^A_\K(\vi), \\
        \Leftrightarrow & \lbr (1+\epsilon)(1-\xi_\K) - (1+\delta) \rbr  \eta^A_\K(\vi) \geq (1+\epsilon)(1-\xi_\K)   \frac{\lambda C\sqrt{d}}{{K}^{d+1}}.
    \end{align*}
    \end{allowdisplaybreaks}
    Because $\vi \notin \mathcal{J}^\K(c_1,{c_2})$, $\eta^A_\K(\vi)\geq c_1 K^{-(D+{c_2})}$. It is enough to let $K$ satisfy that 
    \begin{align}
        & \lbr (1+\epsilon)(1-\xi_\K) - (1+\delta) \rbr   c_1 {K}^{-(d+{c_2})} \geq (1+\epsilon)(1-\xi_\K)   \frac{\lambda C\sqrt{d}}{K^{d+1}}, \nonumber\\
        \Leftrightarrow & \lbr (1+\epsilon)(1-\xi_\K) - (1+\delta) \rbr   c_1  \geq (1+\epsilon)(1-\xi_\K)   \lambda C \sqrt{d} K^{{c_2}-1}. \label{equ: two sides hold then DMD}
    \end{align}
    When $K$ goes to infinity, the limit of the L.H.S. of \eqref{equ: two sides hold then DMD} is $c_1(\epsilon - \delta)$, and the limit of the R.H.S. of \eqref{equ: two sides hold then DMD} is $0$, because $c_2$ is chosen in $(0,1)$. Therefore, the discrete measure dominance holds also for $\vi \in \mathcal{D}^\K$ when $K$ is large enough. In particular, if we choose $c_1 =1$ and $c_2 = \frac{1}{2}$, then it suffices to let 
    \[ K \geq \left( \frac{(1+\delta)(1+\epsilon)\lpar Ld+ \frac{2+\delta}{1+\delta}C\sqrt{d} + 1\rpar }{\epsilon-\delta}\right)^2.\]
    Because $\delta$ can be arbitrarily small, we can let 
    \[K = \epsilon^{-2} [(1+\epsilon)(Ld+2C\sqrt{d}+1)]^2 = \Theta(\epsilon^{-2}).\]
\end{proof}

From the above two lemmas, \cref{thm: T-O of K-MW given Lipschitz condition} follows immediately that the $\K$-MW family is throughput-optimal for any resource requirement $V$ with Lipschitz continuous p.d.f. $f_V$.
\begin{reptheorem}{thm: T-O of K-MW given Lipschitz condition}[Throughput-optimality of $\K$-MW]
    In the continuous MRJ setting with exponential durations, for any resource requirement distribution $V$ with Lipschitz continuous probability density function $f_V$ over $(0,1]^d$, the family of $\K$-MW policies is throughput-optimal. That is, if the continuous MRJ system is stable under some policy $\pi$, then there exists some $\K$ such that $\K$-MW stabilizes the system. 
\end{reptheorem}

\begin{proof}[Proof of \cref{thm: T-O of K-MW given Lipschitz condition}]
    For any system that can be stabilized, with Lipschitz continuous $f_V$, \cref{lem: Lipschitz condition induces discretized dominance} shows the existence of $\K = K\vec 1$ and some service option distribution $\beta_\K$, such that the corresponding discretized service measure satisfies discrete measure dominance \eqref{equ: condition discrete dominance mu and nu}. This is a sufficient condition for the stability of $\K$-MW, as proven in \cref{lem: suff and necc cond for stability}. 
\end{proof}

\subsection{Proof of \cref{thm: MSR+K-MW}: Throughput-Optimality of the $\K$-nMSR Family}
\label{ssec: Proof of T-O KnMSR}

We next show \cref{thm: MSR+K-MW}, the throughput-optimality of the $\K$-discretized nMSR ($\K$-nMSR) family under the condition of Lipschitz continuous $f_V$. We first refer to one of the main results in \citet{chen2024analyzing}, which is a stability condition for nonpreemptive Markovian Service Rate (nMSR).

\begin{lemma}[Throughput-optimality of nMSR. Theorem 4 in \cite{chen2024analyzing}, with notation modified.]
\label{lem: Chen 2024 analyzing T-O nMSR}
    The class of nMSR policies is throughput-optimal. That is, if there exists a policy that stabilizes a multiresource-job system with $J$ job types and arrival vector $\Lambda_J\in \R^J$, then there exists an nMSR policy, $\pi$, with $J^\pi \leq J$ working states that stabilizes the system. Specifically, an nMSR policy can stabilize a system with nonpreemptible jobs if and only if there are some $\epsilon>0$ and  
    \begin{equation}
        \label{equ: T-O nMSR Chen}
        (1+\epsilon)\Lambda_J \in ConvH(C^{(J)}),
    \end{equation}
    where $C^{(J)}$ denotes the set of all possible schedules.
\end{lemma}

Now we prove \cref{thm: MSR+K-MW}, by combining \cref{lem: Lipschitz condition induces discretized dominance} and \cref{lem: Chen 2024 analyzing T-O nMSR}.

\begin{reptheorem}{thm: MSR+K-MW}[Throughput-optimality of $\K$-nMSR]
    In the continuous MRJ setting, the family of $\K$-nMSR policies is throughput-optimal for all $V$ with Lipschitz continuous density $f_V$.
\end{reptheorem}

\begin{proof}[Proof of \cref{thm: MSR+K-MW}]
    In the context of this work, we have $\Lambda_J = \Lambda_\K$, $J = \prod_{l=1}^d K_l$, and $C^{(J)} = C_\K$. In addition, \cref{lem: Lipschitz condition induces discretized dominance} implies that when $f_V$ is Lipschitz continuous, we can always find some $\K$ such that the discrete measure dominance condition \eqref{equ: condition discrete dominance mu and nu} holds. The discrete measure dominance condition is equivalent to the convex hull condition \eqref{equ: condition Lambda in convex hull of C_K}, as shown in \cref{lem: equivalent suff and neec conds}. This means the stability condition \eqref{equ: T-O nMSR Chen} of $\K$-nMSR in \cref{lem: Chen 2024 analyzing T-O nMSR} is satisfied. Therefore we conclude the stability of the CTMC $\{q_t\}_{t\geq 0}$ and the positive recurrence of the empty state. 
\end{proof}

\subsection{Proof of \cref{thm: vertices KEMW and KEnMSR}: Throughput-Optimality of the Extreme-vertices $\K$-EMW and $\K$-EnMSR Families}
\label{ssec: Proof of T-O Vertices Version}

Before we prove \cref{thm: vertices KEMW and KEnMSR}, we provide the following auxiliary lemma, which gives a necessary condition for a stable $\K$-discretized policy.

\begin{lemma}[Necessary Condition for Stable $\K$-discretized Policy]
    \label{lem: nec cond of stability}
    In the continuous MRJ setting, for a fixed $\K \in \N^d$, if there exists any $\K$-discretized policy that stabilizes the system, then there exists a service option distribution $\beta_\K$ over $C_\K$ such that its corresponding service measure satisfies the discrete measure dominance \eqref{equ: condition discrete dominance mu and nu}.
\end{lemma}

\begin{proof}
    Assume that the queueing system is stable under some $\K$-discretized policy $\pi$, i.e., the empty state of the CTMP $\emptyset \in \mathcal S$ is positive recurrent. Then considering the CTMC $\{q( t )\}_{t\geq 0}$, we have that the empty state $ \vec {0}:=(0)_{\Gamma_{\K}}$ is also positively recurrent. Because the CTMC is irreducible, it is positive recurrent and thus ergodic. Hence, an ergodic distribution over the state space $\N^{\Gamma_\K}$ induces a stationary distribution $\beta'_\K$ over the set of all schedulable service options $\pi(\N^{\Gamma_\K}) \subset C_{\vec K}$, such that 
    \begin{enumerate}
        \item The probability $\epsilon$ that the server serves no job is identical to the ergodic measure of the empty state $\vec {0}$, which is positive: $\epsilon = \beta'_\K({\vec  0}) >0.$  
        \item The service rate vector under $\beta'_\K$ matches the arrival rate vector such that the equilibrium is achieved: $\E^{\beta'_\K} M = \Lambda_{\vec K}$.
    \end{enumerate}
    Therefore, when we let $\beta_\K$ be the distribution of service options over $C_\K$ such that $\beta(M) = \frac{ \beta'_\K (M)}{1-\epsilon}$ for all $M\neq \vec 0 $, the measure dominance \eqref{equ: condition discrete dominance mu and nu} holds via the following inequality:
    \begin{align*}
        &\sum\limits_{M\in C_{\vec K}}  \beta'_\K (M) M  = \sum\limits_{M\in C_{\vec K}\setminus \vec  0}  \beta'_\K (M) M = (1-\epsilon)\sum\limits_{M\in C_{\vec K}\setminus \vec {0}}\frac{ \beta'_\K (M)}{1-\epsilon}  M = \Lambda_{\vec K}\\
        \Rightarrow & \sum\limits_{M\in C_{\vec K}\setminus \vec {0}} \beta_\K(M) M =  \sum\limits_{M\in C_{\vec K}\setminus \vec {0}}\frac{\nu(M)}{1-\epsilon}  M = \frac{1}{1-\epsilon} \Lambda_{\vec K} \succ  \Lambda_{\vec K}. \quad\quad\quad \quad\quad\quad \quad\quad\quad \qedhere
    \end{align*}
\end{proof}

\begin{remark}[Sufficient and Necessary Condition for $\K$-MW]
    \label{rmk: suff and necc cond}
    \cref{lem: suff and necc cond for stability} shows that the discrete measure dominance condition \eqref{equ: condition discrete dominance mu and nu} is sufficient for the stability of $\K$-MW. \cref{lem: equivalent suff and neec conds} shows that the convex hull condition \eqref{equ: condition Lambda in convex hull of C_K} is an equivalent condition of \eqref{equ: condition discrete dominance mu and nu}, for the stability of $\K$-MW. \cref{lem: nec cond of stability} shows either conditions \eqref{equ: condition discrete dominance mu and nu} and \eqref{equ: condition Lambda in convex hull of C_K} are also necessary.
\end{remark}

\cref{thm: vertices KEMW and KEnMSR}, the throughput-optimality of the X $\K$-EMW and X $\K$-EnMSR families given Lipschitz continuous p.d.f. $f_V$ of the resource requirement $V$, can be proven directly using \cref{lem: suff and necc cond for stability,lem: Lipschitz condition induces discretized dominance,lem: Chen 2024 analyzing T-O nMSR,lem: nec cond of stability}, and \cref{rmk: suff and necc cond}.

\begin{reptheorem}{thm: vertices KEMW and KEnMSR}[Throughput-optimality of X $\K$-EMW and X $\K$-EnMSR]
    In the continuous MRJ setting, given the same conditions as in \cref{thm: T-O of K-MW given Lipschitz condition}, X $\K$-EMW and X $\K$-EnMSR are both throughput-optimal. Moreover, for any $\K \in \N^d$ fixed, $\K$-MW, X $\K$-EMW, $\K$-nMSR, and X $\K$-EnMSR are all  stable for the same set of arrival rates.
\end{reptheorem}

\begin{proof}[Proof of \cref{thm: vertices KEMW and KEnMSR}]
    Recall that $E^X_\K$ is the set of extreme vertices of $C_\K$. Therefore, $E^X_\K$is the set of the extreme points of the polyhedron $ConvH(C_\K)$, and $ConvH(C_\K) = ConvH(E^X_\K)$. 
    
    By the convex hull condition \eqref{equ: condition Lambda in convex hull of C_K} provided in \cref{lem: equivalent suff and neec conds} for the stability of the system, we conclude the sufficient and necessary condition for the stability of X $\K$-EMW coincides with \eqref{equ: condition discrete dominance mu and nu}. Hence, X $\K$-EMW and $\K$-MW have identical stability regions. The throughput-optimality of X $\K$-EMW is then implied directly from the throughput-optimality of $\K$-MW by \cref{thm: T-O of K-MW given Lipschitz condition}.

    We note that the stability condition of $\K$-nMSR \eqref{equ: T-O nMSR Chen} in \cref{lem: Chen 2024 analyzing T-O nMSR} is equivalent to the stability condition \eqref{equ: condition Lambda in convex hull of C_K}. Also because $ConvH(C_\K) = ConvH(E^X_\K)$, we conclude that X $\K$-EnMSR has the same stability region as $\K$-nMSR. The throughput-optimality of X $\K$-EnMSR likewise holds.
\end{proof}

\begin{remark}[XP $K$-EMW and XP $K$-EMW-B]
    \label{rmk: XP K-EMW and XP K-EMW are TO}
    Recall that in \cref{sssec: Efficient discretized MaxWeight}, we introduced the Pairwise Extreme-vertices (XP) efficient set $E^{\mathrm{XP}}_K$ defined by \eqref{equ: XP efficient set}. 
    For any distinct $M_1,M_2\in \tilde C_{K/2}$, the service option $M_1+M_2$ can be written as a nontrivial convex combination of two distinct service options $2M_1,2M_2\in C_K$:
    $M_1+M_2=\frac{1}{2}(2M_1)+\frac{1}{2}(2M_2).$ Therefore, we guarantee that $E^{\mathrm{XP}}_K$ removes only service options that are not extreme points of $C_K$. Hence, $E^X_K \subset E^{\text{XP}}_K \subseteq  \tilde C_K \subseteq C_K$, and $ConvH(E^{\text{XP}}_K) = ConvH(C_K)$.
    By \cref{lem: equivalent suff and neec conds}, the XP $K$-EMW family is throughput-optimal in the MSJ setting when the resource requirement distribution has a Lipschitz continuous p.d.f. In addition, by \cref{cor: Backfilling Maintains Stability}, we also ensure the throughput-optimality of the XP $K$-EMW-B family under the above condition.
    Because no efficient algorithm is known for constructing $E^X_K$ \cite{shlyk2023number}, $E^{\mathrm{XP}}_K$ provides a tractable proxy for $E^X_K$ in practice. Empirically, the ratio $|E^{\mathrm{XP}}_K|/|C_K|$ also converges to zero as $K$ increases.
\end{remark}

\subsection{Necessary Stability Condition for an MRJ System}
\label{ssec: Necessary Stability Condition for an MRJ System: superset}
The following lemma provides an upper bound for the stability region of any policy. We will use this lemma to prove \cref{thm: 2DK-EMW decreasing,thm: 2JK-EMW Unif} in the \cref{ssec: Proof of Decreasing 2D Version,ssec: Proof of Unif 2J Version}.

\begin{lemma}
    \label{lem: MRJ stability region subset}
    Consider an MRJ system with resource requirement distribution $V = (V_1,...,V_d)$. We let $m$ be the resource type that has the largest mean requirement, i.e., $m: = \arg\max\limits_{1\leq j \leq d} \E V_j$.
    Then the stability region for any policy is a subset of 
    \begin{equation}
        \label{equ: stability region super set}
        \{ \lambda>0: \lambda \E V_m <1\}.
    \end{equation}
\end{lemma}

\begin{proof}
    We choose the test function $w_m: \mathcal S \rightarrow \R_{+,0}$ to be the total requirement of the type-$m$ resource in the system, i.e., workload of type-$m$ resource. Formally,  when the state is given by $x = (\vec v_n)_{1\leq n\leq N}$, $\phi(x) = \sum_{n=1}^N \vec v_n(m)$. Assume, for the sake of contradiction, that when $\lambda \geq \frac{1}{\E V_m }$, there exists a policy $\pi$ stabilizing the system and resulting in the steady state distribution $X$. Then the ergodic probability at the empty state $\pro(X = \emptyset) >0$.
    
    According to \citet{glynn2008bounding}, at the steady state $X$, we have zero expected drift of $w_m(X)$ under the policy $\pi$:
    \begin{equation}
        \label{equ: steady state drift is zero}
        \E [\Delta_\pi \circ w_m(X)] = 0.
    \end{equation}
    Let $H$ denote the stationary total resource requirement in service. As a result, the rate of decrease of $w_m(x)$ is $H$, where we recall that jobs have duration distribution $\text{Exp}(1)$. On the other hand, the rate of increase of $\phi$ due to arrivals is $\lambda \E V_m$. Therefore, we can express the expected drift explicitly, such that
    \[
    \E [\Delta_\pi \circ w_m(X)] = \lambda \E V_m - \E[H|X\neq \emptyset]\pro(X = \emptyset).
    \]
    Because $H\leq 1$, 
    \[\E [\Delta_\pi \circ w_m(X)] > \lambda \E V_m - 1 \geq 0,\]
    which contradicts the fact of zero expected drift \eqref{equ: steady state drift is zero}. Therefore, we conclude that the stability region of any policy is a subset of $\{\lambda:  \lambda \E V_m  <1\}$.
\end{proof}

\subsection{Proof of \cref{thm: 2DK-EMW decreasing}: Throughput-Optimality of the 2-Bucket $K$-EMW and $\K$-EnMSR Families}
\label{ssec: Proof of Decreasing 2D Version}

We now prove \cref{thm: 2DK-EMW decreasing}, the throughput-optimality of the $K$-EMW and $K$-EnMSR families given $f_V$ is weakly decreasing in $(0,1]$. 

\begin{reptheorem}{thm: 2DK-EMW decreasing}[2B $K$-EMW and 2B $K$-EnMSR]
    In the continuous MSJ setting with the weakly decreasing density function of $V$, the 2B $K$-EMW and 2B $K$-EnMSR families are throughput-optimal. In particular, for any $\lambda <\frac{1}{\E V}$, there exists some $K$ such that the queuing system is stable under 2B $K$-EMW and under 2B $K$-EnMSR. In particular, it is enough to let $K=2^L$, where $L = \lfloor -\log_2\lpar \frac{1}{\lambda}-\E V \rpar \rfloor+1$.
\end{reptheorem}

\begin{proof}[Proof of \cref{thm: 2DK-EMW decreasing}]
    By \cref{lem: MRJ stability region subset}, note that only if $\lambda < \frac{1}{\E V}$ can a system be stabilized by some policy.
    We then move to show that for any $\lambda < \frac{1}{\E V}$, we can always find some $K = 2^L$ such that 2B $K$-EMW is stable.

    After we fix some integer $L$ such that $K = 2^L$, we conduct a $K$-discretization as derived in \eqref{equ: K-partition}. We recall that in \cref{sssec: Efficient discretized MaxWeight}, the service option is indexed by the type of the larger job it attempts to serve:
    \begin{itemize}
        \item For $k = 2^\ell\in [K]$ that is a power-of-two integer (including $k=1$), we include the service option that serves $2^{L-\ell}$ type-$k$ jobs simultaneously. Specifically, $M_{2^\ell} = \{(2^\ell, 2^{L-\ell})\}$.
        \item For $k \in [K]$ that is not a power-of-two integer, let $\ell:=\ell(k) = \lceil \log_2 j \rceil$. We include the service option that serves $2^{L-\ell}$ type-$k$ jobs and $2^{L-\ell}$ type-$(2^\ell-k)$ jobs simultaneously. Specifically, $M_k = \{(k,2^{L-\ell}),(2^\ell - k, 2^{L-\ell})\}$.
    \end{itemize}
        
    Our remaining task is to assign appropriate probability mass $\beta_{k}:=\beta(M_{k})$ to each $M_{k}$, such that the corresponding discrete service measure $\eta_K^S$ satisfies $\eta_K^S( k ) >\eta^A_K( k )$ for all $1\leq k \leq K=2^L$. As the stability conditions for 2B $K$-EMW and 2B $K$-EnMSR are the same, the discrete measure dominance condition \eqref{equ: condition discrete dominance mu and nu} suffices to show both policies are stable. 

    To establish a service option distribution $\beta_K$, we first construct $L+1$ sequences of nonnegative real numbers $\{p_j^{(\ell)}\}_{1\leq j \leq 2^{L-\ell}}$, where sequences $\ell = 0,\cdots,L$ are constructed iteratively. In round $0$, $\{p_j^{(0)}\}_{1\leq j \leq 2^{L}}$ is defined based on the requirement distribution vector $(p_k)_{1\leq k\leq K}$, where $p_k : = \pro(V\in I_k)$. Then, in the $r$-th round, we construct $\{p_j^{(r)}\}_{1\leq j \leq 2^{L-r}}$ recursively based on $\{p_j^{(r -1)}\}_{1\leq j \leq 2^{L-r +1}}$ when $r \geq 1$. Intuitively, for each round $r$, the arrival probability type-$j$ jobs with $j > 2^{L-r}$ are fully covered by service option $M_k$, and the sequence $p_j^{(r)}$ ($1\leq j \leq 2^{L-r}$) represents the fraction of type-$j$ jobs not covered by service option $M_k$, with $k > 2^{L-r}$. Additionally, in the $r$-th round, for each uncovered job type $j > 2^{L-r-1}$, we note that only $M_j$ is able to serve some type-$j$ jobs, among the remaining service options $\{M_k\}_{k\leq 2^{L-r}}$. Therefore, for each round $0\leq r\leq L$ and each corresponding job type $2^{L-r-1}< j \leq 2^{L-r}$, we assign $M_j$ some mass $\beta_j$ proportional to $p_j^{(r)}$ to cover the remaining fraction of type-$j$ jobs. These proxies of masses $p_j^{(r)}$ then form a sequence of nonnegative numbers $\{p^{(L-\ell(k))}_k\}_{1\leq k \leq K}$. After the construction of this a mass-proxy sequence $\{p^{(L-\ell(k))}_k\}_{1\leq k \leq K}$, we will set the probability mass $\beta_k$ assigned to $M_k$ to be $\frac{(1+\epsilon)\lambda}{2^{L-\ell(k)}} p^{(L-\ell(k))}_k$, for some $\epsilon$ specified later. We will show that the discrete service measure condition induced by our service option distribution $\{\beta_k\}_{1\leq k\leq K}$ is satisfied. Moreover, we will demonstrate that the masses $\beta_{k}$ sum to at most $1$ and hence the $\beta_K$ is a valid probability distribution.

    Now, we begin our $0$-th round by simply letting $p^{(0)}_j = p_j$, for all $1\leq j \leq L$. Because all masses $p_j$ are nonnegative, $p^{(0)}_j \geq 0$ directly holds for each $2^{L-1}<j<2^L$. We recall that we only take $\{p^{(0)}_j\}_{j>2^{L-1} }$ into our mass-proxy sequence $\{p^{(L-\ell(k))}_k\}_{1\leq k \leq K}$, but we will use the whole set $\{p^{(0)}_j\}_{1\leq j\leq 2^L }$ in our following analysis.

    In the $1$-st round, we consider the uncovered job types $1\leq j \leq 2^{L-1}$. We define $p^{(1)}_j := p^{(0)}_{2^L-j} - p^{(0)}_{j}$ be the difference of the complementary pair, where $1\leq j < 2^{L-1}$, and define $p^{(1)}_{2^{L-1}} := p^{(0)}_{2^L-1}$ that copies its value in the previous round. Because $ p^{(0)}_{j} = p_{j}$ is weakly decreasing with respect to its index $j$, we have the following inequality hold for each $1\leq j\leq j' \leq 2^{L-1}$,
    \[ p^{(0)}_{2^L-j} \geq p^{(0)}_{2^L-j'}  \geq  p^{(0)}_{j'} \geq  p^{(0)}_{j},\]
    which implies $p^{(1)}_{j} \geq  p^{(1)}_{j'} \geq 0$. Therefore, we have that $p^{(1)}_{j}$ is also weakly decreasing. Thereafter, in round $r$, the weakly decreasing property does not include the boundary index $2^{L-r}$, because we do not need it to subtract an element in the next round but just copy it, so that its non-negativity always holds. 

    We continue the above iterations, and in the $r$-th round, we have a nonnegative sequence $\{p^{(r-1)}_j\}_{j\leq 2^{L-r+1}}$ obtained in the previous round, of which $\{p^{(r-1)}_j\}_{j< 2^{L-r+1}}$ is weakly decreasing. 
    By defining $p^{(r)}_j := p^{(r-1)}_{2^{L-r+1}-j} - p^{(r-1)}_{j}$ for $1\leq j < 2^{L-r}$, and copying $p^{(r)}_{2^{L-r}} := p^{(r-1)}_{2^{L-r}}$, we obtain a new sequence $\{p^{(r)}_j\}_{1\leq j \leq 2^{L-r}}$. This sequence is weakly decreasing (excluding the boundary element $p^{(r)}_{2^{L-r}}$) and nonnegative, due to the weakly decreasing property of $\{p^{(r-1)}_j\}_{j< 2^{L-r+1}}$. Finally we reach the $(L-1)$-th round and have only $p^{(L-2)}_{2}$ and $p^{(L-2)}_{1}$ from the $(L-2)$-th round. We notice that there is no complementary index pair, and therefore define $p^{(L-1)}_{2} := p^{(L-1)}_{2}$ and directly move to the $L$-th round, where it remains only to define $p^{(L)}_{1} := p^{(L-2)}_{1}$.

    Note that our sequence $\{p^{(L-\ell(j))}_j\}_{1\leq j \leq K}$ use only $K$ elements in $\{p^{(r)}_j\}_{1\leq j \leq 2^{L-r},0\leq r\leq L}$. We let $\epsilon>0$ be a real number whose value will be specified later, and assign the probability mass $\beta_k := \frac{(1+\epsilon)\lambda}{2^{L-\ell(k)}} p^{(L-\ell(k))}_k$ to the service option $M_k$.
    
    Before we show that $\{\beta_k\}_{k=1}^K$ is a valid probability measure, we prove that the discrete measure dominance \eqref{equ: condition discrete dominance mu and nu} holds. For each non-power-of-two integer $k$, the service option that attempt to serve some type-$k$ job(s) are $M_{2^\ell - k}$, $\ell> \ell(k)$, and $M_k$. Service option $M_{2^\ell - k}$ attempts to serve $2^{L-\ell}$ type-$k$ jobs and service option $M_k$ attempts to serve $2^{L-\ell(k)}$ jobs. Therefore, the discrete service measure $\eta^S_K(k)$ satisfies that
    \begin{align}
        \eta^S_K(k) & = 2^{L-\ell(k)}\beta_k  + \sum\limits_{\ell=\ell(k)}^{L} 2^{L-\ell} \beta_{2^\ell - k} \nonumber\\
        & = 2^{L-\ell(k)}\frac{(1+\epsilon)\lambda}{2^{L-\ell(k)}} p^{(L-\ell(k))}_k + \sum\limits_{\ell=\ell(k)}^{L} 2^{L-\ell}\frac{(1+\epsilon)\lambda}{2^{L-\ell}} p^{(L-\ell)}_{2^\ell-k}\nonumber\\
        & = (1+\epsilon)\lambda \Big[ p^{(L-\ell(k))}_k + \sum\limits_{\ell=\ell(k)}^{L} p^{(L-\ell)}_{2^\ell-k}  \Big] \nonumber\\
        & = (1+\epsilon)\lambda p_k \label{equ:iterative relationship} \\
        & > \eta^A_K(k)\nonumber,
    \end{align}
    where \eqref{equ:iterative relationship} is due to the iterative decomposition $p^{(r)}_j = p^{(r-1)}_{2^{L-r+1}-j} - p^{(r-1)}_{j}$. For each power-of-two index $k $, the service option that attempt to serve some type-$k$ job(s) are $M_{2^\ell - k}$, $\ell \geq \ell(k)+1$, and $M_k$. Likewise, using the fact that $p^{(L-\ell(k))}_{k} := p^{(L-\ell(k)+1)}_{k}$ in addition to the above argument, we obtain that $\eta^S_K(k) = (1+\epsilon)\lambda p_k > \eta^A_K(k)$, and conclude discrete measure dominance.

    We next show that there exists $L$, $\epsilon$, such that $\beta_K$ is a valid probability distribution, by arguing that $\sum_{1\leq j\leq 2^L}\beta_j \leq (1+\epsilon) \lambda \E V +o(2^{-L})$. To achieve this, we first claim that the expectation of $V$ can be expressed as
    \begin{equation}
        \label{equ: where does EV come from?}
        \E V = \sum\limits_{j=1}^{2^L} \Big[\frac{\beta_{j}}{(1+\epsilon )\lambda } \sum\limits_{k=1}^{2^L} M^{( k )}_{j} \E[V|V\in I_k] \Big].
    \end{equation}
    Equation \eqref{equ: where does EV come from?} holds because 
    \[ \frac{\beta_{j}}{(1+\epsilon )\lambda } = 2^{-L+\ell(j)} p^{(L-\ell(j))}_{j},  \]
    and $M^{(k)}_{j} = 2^{L-\ell(j)}$ if and only if $M_{j}$ serves at least one jobs in the bucket $k$. Therefore, the R.H.S. of \eqref{equ: where does EV come from?} can be written as 
    \begin{allowdisplaybreaks}[0]
     \begin{align}
        & \sum\limits_{j=1}^{2^L} \Big[p^{(L-\ell(j))}_{j} \sum\limits_{k=1}^{2^L} \one \{M^{( k )}_{j} >0\} \, \E[V|V\in I_k] \Big]\nonumber \\
        &\quad \quad = \sum\limits_{k=1}^{2^L} \bigg[ \Big( \sum\limits_{j=1}^{2^L} p^{(L-\ell(j))}_{j} \one \{M^{( k )}_{j} >0\}  \Big)\, \E[V|V\in I_k]  \bigg] \nonumber\\
        &\quad \quad  = \sum\limits_{k=1}^{2^L} \big(p_{k} \, \E[V|V\in I_k] \big) \label{equ: severing rates in one interval sum up to p}\\
        &\quad \quad = \E V, \nonumber
    \end{align}   
    \end{allowdisplaybreaks}
    where step \eqref{equ: severing rates in one interval sum up to p} is due to the iterative decomposition $p^{(r)}_j = p^{(r-1)}_{2^{L-r+1}-j} - p^{(r-1)}_{j}$ that we also applied to justify \eqref{equ:iterative relationship}, or alternatively, because the summed serving rates at interval $I_k$ is $(1+\epsilon)\lambda p_k$:
    \begin{equation}
        \label{equ: the summed serving rates at interval I_k is p_k}
        \sum\limits_{j=1}^{2^L} p^{(L-\ell(j))}_{j} \one \{M^{( k )}_{j} >0\}  = p^{(k)}.
    \end{equation}

    After decomposing $\E V$, we likewise decompose the summed mass $\sum_{1\leq j\leq 2^L}\frac{\beta_j}{(1+\epsilon) \lambda}$.
    We use the property that for any $j$, the service option $M_{j}$ does not waste any capacity, i.e., $\sum\limits_{k=1}^{2^L} M^{( k )}_{j} \frac{k}{2^L} =1$. Therefore,
    \[
        \sum\limits_{j=1}^{2^L}   \Big[\frac{\beta_{j}}{(1+\epsilon )\lambda } \Big(\sum\limits_{k=1}^{2^L} M^{( k )}_{j} \frac{k}{2^L} \Big)\Big]
         \;  = \; \sum\limits_{j=1}^{2^L} \frac{\beta_{j}}{(1+\epsilon )\lambda }.
    \]
    In addition, $\frac{k-1}{2^L}\leq \E[V|V\in I_k] \leq \frac{k}{2^L}$ means the waste of capacity by up-rounding each job's requirement is controlled: $\frac{k}{2^L}-\E[V|V\in I_k]  \leq 2^{-L}$. This result leads to the following bound:
    \begin{allowdisplaybreaks}[0]
    \begin{align}
        &\sum\limits_{j=1}^{2^L} \frac{\beta_{j}}{(1+\epsilon )\lambda } -  \E V \nonumber\\
        &\quad \quad = \sum\limits_{j=1}^{2^L} \bigg[\frac{\beta_{j} }{(1+\epsilon )\lambda } \sum\limits_{k=1}^{2^L}  \Big[ M^{( k )}_{j}  \Big(\frac{k}{2^L} -\E[V|V\in I_k]\Big) \Big] \bigg]\label{equ: due to lemma where EV}\\
        & \quad \quad \leq \sum\limits_{j=1}^{2^L} \bigg[\frac{\beta_{j} }{(1+\epsilon )\lambda } \sum\limits_{k=1}^{2^L}  \Big( M^{( k )}_{j}  \frac{1}{2^L}  \Big) \bigg] \nonumber\\
        & \quad \quad = \frac{1}{2^L} \sum\limits_{j=1}^{2^L} \Big(\sum\limits_{k=1}^{2^L} \one \{M^{( k )}_{j}) >0 \} \, p^{(L-\ell(j))}_{j} \Big)  \nonumber\\
        & \quad \quad =  \frac{1}{2^L}, \label{equ: sum pjk = sum p1 =1}
    \end{align}
    \end{allowdisplaybreaks}
    where step \eqref{equ: due to lemma where EV} is due to the decomposition of $\E V$ \eqref{equ: where does EV come from?}, and step \eqref{equ: sum pjk = sum p1 =1} is again due to \eqref{equ: the summed serving rates at interval I_k is p_k}.
    
    It now remains to bound the summed mass. By letting $L = \lceil -\log_2\lpar \frac{1}{(1+\epsilon)\lambda}-\E V \rpar \rceil$ (and thus $K = 2^L = \Theta(\frac{1}{1-\lambda \E V})$), we have
    \[\sum\limits_{j=1}^{2^L} \beta_{j} \leq  (1+\epsilon )\lambda  \E V + \frac{(1+\epsilon )\lambda }{2^L} \leq 1. \]
    As $\epsilon$ can be chosen arbitrarily small, we can guarantee a valid service option distribution $\beta_K$ with $L = \lfloor -\log_2\lpar \frac{1}{\lambda}-\E V \rpar \rfloor+1$. Because we have shown that this service option distribution leads to the discrete measure dominance, we finally conclude \cref{thm: 2DK-EMW decreasing}.
\end{proof}

\subsection{Proof of \cref{thm: 2JK-EMW Unif}: Throughput-Optimality of the 2-Job $\K$-EMW and $\K$-EnMSR Families}
\label{ssec: Proof of Unif 2J Version}

We now prove \cref{thm: 2JK-EMW Unif}, the throughput-optimality of the 2J $\K$-EMW and 2J $\K$-EnMSR families when the p.d.f. $f_V$ of the resource requirement distribution $V$ is centrally symmetric with respect to $\frac{1}{2} \vec 1$ and Lipschitz continuous.

\begin{reptheorem}{thm: 2JK-EMW Unif}[2J $\K$-EMW and 2J $\K$-EnMSR]
    In the continuous MRJ setting, assume that the p.d.f. $f_V$ is Lipschitz continuous and centrally symmetric with respect to $\frac{1}{2} \vec 1$. In this case, both the 2J $\K$-EMW and 2J $\K$-EnMSR families are throughput-optimal. To be specific, for any arrival rate $\lambda <2$, there exists some $\K$ such that both 2J $\K$-EMW and 2J $\K$-EnMSR stabilize the system. 
     
    Moreover, if $V \sim Unif(( 0,1]^d)$, any $\K \succeq K \vec 1$ stabilizes the system, where $K$ is the smallest odd integer such that $K \geq \lfloor \frac{2\lambda d}{2-\lambda} \rfloor+1$. Additionally, in the single-resource setting, it suffices to choose an odd $K$ such that $K\geq \lfloor \frac{\lambda}{2-\lambda}\rfloor +1$.
\end{reptheorem}
\begin{proof}[Proof of \cref{thm: 2JK-EMW Unif}]
    By \cref{lem: MRJ stability region subset}, we note that any system that can be stabilized should satisfy $\lambda <2$.
    We next prove that for any $\lambda<2$, we can find some $\K$ such that our 2J $\K$-EMW stabilizes the system.
    Our goal is to construct a service option distribution $\beta_\K$ satisfying discrete measure dominance \eqref{equ: condition discrete dominance mu and nu} for some $\K$, and use \cref{lem: suff and necc cond for stability} to demonstrate the stability given any $\lambda <2$. 
    We start by considering the general Lipschitz and centrally symmetric case, then we prove the special case result when $V$ is uniform. 

    Our proof for the general case follows a similar breakdown into the exceptional and default sets as \cref{lem: Lipschitz condition induces discretized dominance}, but our only exceptional job types are the boundary job types in $\mathcal B^\K$.
    
    We consider $\K = K \vec 1$ for some odd $K$. Recall in \cref{sssec: Efficient discretized MaxWeight} that we label the service option that serves only one boundary-type job $M_\vi$, if the job is in the boundary set $\vi\in \mathcal B^\K$, i.e., $M_\vi = \{\vi\}$. For $\vi \notin \mathcal B^\K$, $M_\vi$ has a different meaning. Apart from the boundary set, we label the service option that serves both one type-$\vi$ job and one type-($\K - \vi$) job simultaneously $M_\vi$, for each $\vi$ such that $1\leq \vi_1 \leq \frac{K-1}{2}$ and $1\leq \vi_m \leq K-1$, where $2\leq m\leq d$. Formally speaking, we let $\mathcal H^\K: = [\frac{K-1}{2}]\times [K-1] \times \cdots \times [K-1]$ denote the set of all job types with type-1 resource requirement at most $\frac{K-1}{2K}$, and let $M_\vi = \{\vi,\K-\vi\}$ for each $\vi \in \mathcal H^\K$. 

    For simplicity, we let $\beta_\vi:=\beta_\K(M_\vi)$ for service option $M_\vi$. We fix some positive $\epsilon$ and begin by assigning the mass $\beta_\vi = (1+\epsilon) \eta^A_\K(\vi)$ to the policy $M_\vi$ for each $\vi \in \mathcal B^\K$, and assign the mass $\beta_\vi = (1+\epsilon) \max\lpar \eta^A_\K(\vi), \eta^A(\K-\vi)\rpar$ for each $\vi \in \mathcal H^\K$. Then the discrete measure dominance \eqref{equ: condition discrete dominance mu and nu} is satisfied.  It remains to show  $\beta_\K$ is a valid probability distribution, i.e., $\sum_{\vi \in \mathcal B^\K \cup \mathcal H^\K} \beta_\vi \leq 1$. We apply \cref{lem: Upper bound for the Supremum of Lipschitz Density Function} using the Lipschitz constant $C$ of $f_V$, so we have $\sup f_V \leq L:= L(C,d)$, and can bound $\eta^S_\K (\mathcal B^\K) = \sum_{\vi \in \mathcal B^\K} \beta_\vi $ such that 
    \[ \eta^S_\K (\mathcal B^\K)\leq L K^{-d} dK^{d-1} =  Ld K^{-1}.\]
    Then, for each non-boundary job type $\vi\in \mathcal H^\K$, we bound the difference in service measure between the two job types served by $M_\vi$. 
    \begin{allowdisplaybreaks}[0]
    \begin{align*}
        |\eta^A_\K(\vi) - \eta^A_\K(\K-\vi)| & = \lambda \left | \int_{I_\vi} - \int_{I_{\K-\vi}}  \right | f_V(x) dx  \\
        & = \lambda \left |  \int_{I_\vi} f_V(x) dx  -   \int_{I_{\K-\vi}} f_V(\vec 1 -x) dx \right |  \\
        & =\lambda \left | \int_{I_\vi} \lpar f_V(x) - f_V(x+K^{-1}\vec 1) \rpar  dx \right | \\
        & \leq  \lambda  \int_{I_\vi} \left | f_V(x) - f_V(x+K^{-1}\vec 1) \right |  dx\\
        & \leq \lambda  \int_{I_\vi} C\sqrt{d}K^{-1} dx \\
        & =  \lambda C\sqrt{d}K^{-d-1},
    \end{align*}
    \end{allowdisplaybreaks}
    As a result, we have $\beta_\vi \leq (1+\epsilon)[ \eta^A_\K(\vi)+ C\sqrt{d}K^{-d-1}]$. Therefore, because the number of non-boundary job types $|\mathcal H^\K| = \frac{1}{2}\lpar \frac{K-1}{K} \rpar^d$, the summed mass is bounded by 
    \begin{align}
        \sum_{\vi \in \mathcal B^\K \cup \mathcal H^\K} \beta_\vi & \leq (1+\epsilon) \lambda \lbr Ld K^{-1} + \sum_{\vi \in \mathcal H^\K} \lpar \frac{1}{\lambda} \eta^A_{\K}(\vi) + C\sqrt{d}K^{-d-1}\rpar \rbr \nonumber\\
        & \leq (1+\epsilon) \lambda \lbr \lpar Ld + \frac{C\sqrt{d}}{2} \lpar \frac{K-1}{K}\rpar^d  \rpar K^{-1} + \frac{1}{\lambda}\sum_{\vi \in \mathcal H^\K}  \eta^A_{\K}(\vi)   \rbr \nonumber\\
        & < (1+\epsilon) \lambda \lbr  (Ld+\frac{C\sqrt{d}}{2})K^{-1} + \frac{1}{2} \rbr. \label{equ: no intersection of HK and its complement}
    \end{align}
    Here, the last inequality \eqref{equ: no intersection of HK and its complement} holds because the two disjoint sets $\mathcal H^\K$ and $\K - \mathcal H^\K$ have identical arrival measure, and their union is a subset of $\Gamma_\K$, which has total arrival measure at most $\lambda$. As a result, $\sum_{\vi \in \mathcal H^\K}  \eta^A_{\K}(\vi)\leq \frac{1}{2}$.
    
    Using the upper bound \eqref{equ: no intersection of HK and its complement}, we choose $K$ to be such that 
    \begin{allowdisplaybreaks}[0]
    \begin{align*}
        & \quad (1+\epsilon) \lambda \lbr (Ld+\frac{C\sqrt{d}}{2})K^{-1} + \frac{1}{2} \rbr \leq 1 \\
        & \Leftrightarrow K \geq  \Big(\frac{1}{(1+\epsilon) \lambda } - \frac{1}{2} \Big)^{-1} \Big( Ld+\frac{C\sqrt{d}}{2}\Big).
    \end{align*}
    \end{allowdisplaybreaks}
    
    This completes the general centrally symmetric case. We now discuss a few special-case results. For the multi-resource uniform setting, $L=1$ and $C=0$. Thus it suffices to let $K \geq \lceil \frac{2d\lambda(1+\epsilon)}{2-(1+\epsilon)\lambda}\rceil$. By letting $\epsilon$ be small enough, we note that $K = \lfloor  \frac{2\lambda d}{2-\lambda}\rfloor+1$ enables discrete measure dominance \eqref{equ: condition discrete dominance mu and nu} to hold.
    
    We then show stability in the uniform single-resource (MSJ) setting. With $M_K = \{K\},M_1 =\{1,K-1\},...,M_{\frac{K-1}{2}}: = \{\frac{K-1}{2}, \frac{K+1}{2}\}$, we let $\beta_K=\beta_1=...,\beta_{\frac{K-1}{2}}= (1+\epsilon)\lambda/K$ for some $\epsilon>0$, so that discrete measure dominance \eqref{equ: condition discrete dominance mu and nu} can be established. It remains to ensure \[\beta_K+\sum\limits_{j=1}^{\frac{K-1}{2}} \beta_j  = (1+\epsilon)\lambda \frac{K+1}{2K}\leq 1 \Rightarrow K \geq \lceil \frac{(1+\epsilon)\lambda} {2-(1+\epsilon)\lambda}\rceil. \] Because $\epsilon$ can be arbitrarily small, it is enough to let $K = \lfloor\frac{\lambda}{2-\lambda} \rfloor+1$.

    Finally, by applying \cref{lem: suff and necc cond for stability,lem: Chen 2024 analyzing T-O nMSR}, the stability of the system is achieved for any $\lambda<2$, respectively under 2J $\K$-EMW and 2J $\K$-EnMSR.
\end{proof}

\section{Empirical Results}
\label{sec: empirical results}

In this section, we use simulations to validate our main results for both parametric distributions and Google Borg trace data \cite{tirmazi2020borg}.
We compare our policies with several index-based benchmarks: FCFS, First-Fit, Best-Fit, and Least-Server-First (LSF). These index-based policies can also be naturally extended to the continuous-resource setting. The detailed definitions of First-Fit, Best-Fit, and LSF are given in Appendix~\ref{sec: Definitions of Some Index-Based Policies}. In particular, we note that LSF is equivalent to naive MaxWeight in the continuous MSJ setting, and very similar in the Google Borg setting. When job resource requirements are continuous and almost surely distinct, each ``job type" contains at most one job. As a result, naive MaxWeight always selects the feasible set containing the largest number of jobs, which is equivalent to batching jobs in increasing order of their resource requirements.

In \cref{ssec: Discretized MaxWeight-typed Policies v.s. index-based Policies: Simple Continuous Resource Requirement}, we simulate MSJ systems with i.i.d. parametric continuous resource requirements, to verify the throughput-optimality of our theoretically-motivated policies. In addition, we demonstrate that with a large discretization parameter $K$, the mean response time of $K$-EMW-B policies is consistently better than all prior index-based policies.

In \cref{ssec: Discretized MaxWeight-typed Policies v.s. index-based Policies: Borg Traces}, we compare our policy with the prior index-based and class-based policies, using resource requirement demand from a Google Borg trace. In this practical setting, although the requirements are not continuous or i.i.d., the performance advantages of our policies generalize beyond our model.

All data and code in this section can be found in the following repository: \url{https://github.com/isaacg1/buckets}.

\subsection{Comparison of Mean Response Time under Parametric Requirement Distribution}
\label{ssec: Discretized MaxWeight-typed Policies v.s. index-based Policies: Simple Continuous Resource Requirement}

We first compare the theoretically-motivated family of policies with index-based policies (FCFS, First-Fit, Best-Fit, and LSF), where the job resource requirements are sampled i.i.d. from some one-dimensional parametric continuous distributions.

We simulate the system with $10^6$ job arrivals. In the simulation, if the queue length ever exceeds $10^4$, or the mean response time exceeds $10^3$, which we represent using dotted lines in our plot, we consider the system to be indistinguishable from unstable via simulation.

Given a known resource requirement distribution $V$, let $D_V$ denote the theoretical stability region, defined as the set of all arrival rates $\lambda > 0$ such that the system is stable under some policy with arrival rate $\lambda$ and distribution $V$. We further define $\lambda_V^* := \sup D_V$ as the stability boundary of $D_V$. In the following experiments, we use the mean response time $\E[T(\lambda)]$ as the performance metric, and visualize them against the system load \[\rho := \frac{\lambda}{\lambda_V^*}.\]

The resource distributions we use in our simulation are uniform (\cref{sssec: Uniform distribution}), truncated normal (\cref{sssec: truncated Normal}), bounded Lomax and triangular (\cref{sssec: Bounded Lomax and Triangle (Weakly Decreasing) Resource Requirement}), and symmetric triangular (\cref{sssec: Sym Tri Resource Requirement}), for which the p.d.f. are constant, centrally symmetric, weakly decreasing, and general Lipschitz, respectively.

\subsubsection{Uniformly Distributed Resource Requirement}
\label{sssec: Uniform distribution}

We first consider the case where job resource requirements follow a uniform distribution on $(0,1]$. Our primary focus is on the 2-Job $K$-discretized Efficient MaxWeight (2J $K$-EMW) policy, introduced in \cref{sssec: Efficient discretized MaxWeight}, which we proved to achieve throughput-optimality in \cref{thm: 2JK-EMW Unif}. In addition, we compare our theoretically-motivated policies with index-based policies, to demonstrate their empirical advantages in reducing the mean response time. Note that the stability boundary is $\lambda_V^* = 2$.

In \cref{fig: uniform-2JEMW}, we compare the mean response times under 2J $K$-EMW, with $K = 4, 8, 16, 32,$ and $64$. Because the policy may leave additional space available, \cref{fig: uniform-2JEMW Back} further compares the performance of 2J $K$-EMW with Backfilling (2J $K$-EMW-B), using the same set of discretization parameters $K$. 

In \cref{fig: uniform-2JEMW}, we observe that for each $K$, the mean response time of 2J $K$-EMW increases when the load $\rho$ grows, and jumps to indistinguishable from unstable. From now on, we call the set of loads for which the policy is demonstrated to be stable the ``empirically-stable" region, which is an empirical lower bound of the theoretical stability region. As $K$ increases, the boundary of the empirically-stable region increases towards $1$, which verifies \cref{thm: 2JK-EMW Unif}. In \cref{fig: uniform-2JEMW Back}, the empirically-stable region of 2J $K$-EMW-B also approaches to $[0,1)$, which is similar to \cref{fig: uniform-2JEMW}.  

In \cref{fig: uniform-2JEMW}, larger values of $K$ result in a worse mean response time until the arrival rate becomes close to the stability boundary. In contrast, in \cref{fig: uniform-2JEMW Back}, when Backfilling is applied, a larger discretization parameter $K$ generally leads to a smaller mean response time. Moreover, comparing the two sets of plots, it is clear that Backfilling consistently yields a smaller mean response time.
We conjecture that this behavior arises because when the load is low and $K$ is large, the system may not always have two complementary jobs. Hence, 2J $K$-EMW may be unable to pack two jobs for joint service, which wastes some capacity. In contrast, Backfilling mitigates this issue, allowing the server to pack non-complementary jobs and possibly additional small jobs. Therefore, in practice, we recommend choosing a larger enough $K$ to ensure stability, as derived in \cref{thm: 2JK-EMW Unif}, but not much larger, and applying Backfilling. As a result, the corresponding policy can achieve stability, computational efficiency, and a better mean response time.

\begin{figure}[ht]
    \centering
    \begin{subfigure}[t]{0.48\linewidth}
        \centering
        \includegraphics[width=\linewidth]{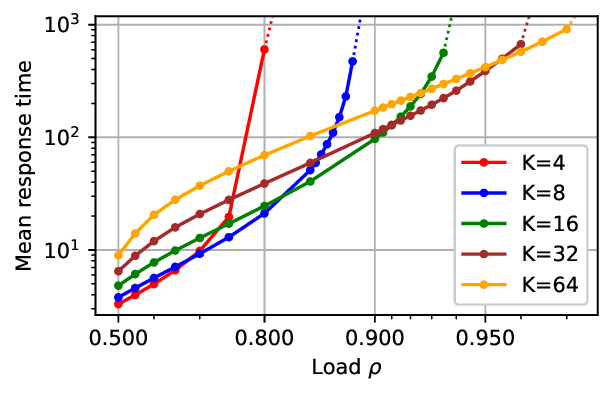}
        \caption{Mean response times under 2J $K$-EMW.}
        \label{fig: uniform-2JEMW}
    \end{subfigure}
    \hfill
    \begin{subfigure}[t]{0.48\linewidth}
        \centering
        \includegraphics[width=\linewidth]{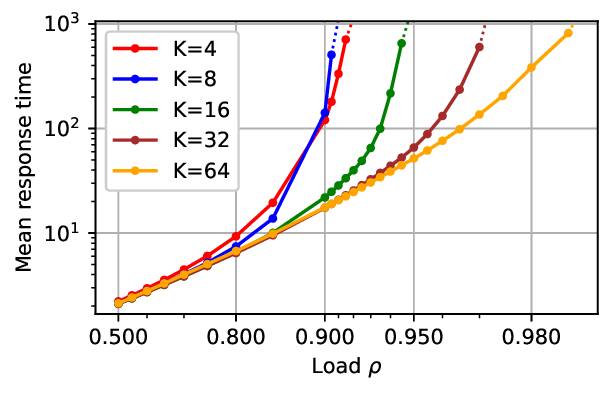}
        \caption{Mean response times under 2J $K$-EMW with Backfilling.}
        \label{fig: uniform-2JEMW Back}
    \end{subfigure}
    \caption{Log-log plots of mean response times as functions of load $\rho$, under 2J $K$-EMW with or without Backfilling, where $K=4,8,16,32,64$ and $V \sim \text{Uniform}((0,1])$. Simulation duration is $10^6$ jobs. }
    \label{fig: Uniform Discretization}
\end{figure}

In \cref{fig: uniform-Comp}, we compare 2J $64$-EMW-B, the best policy in \cref{fig: uniform-2JEMW Back}, with index-based policies (FCFS, First-Fit, Best-Fit, and LSF). The simulation results show that naive FCFS, Best-Fit, and LSF have small empirically-stable regions, whereas First-Fit achieves a stability boundary close to $\rho =1$. Among the policies compared in \cref{fig: uniform-Comp}, 2J $64$-EMW-B achieves the largest empirically-stable region and the least mean response time under high load.

\begin{figure}[ht]
    \centering
    \includegraphics[width=0.6\linewidth]{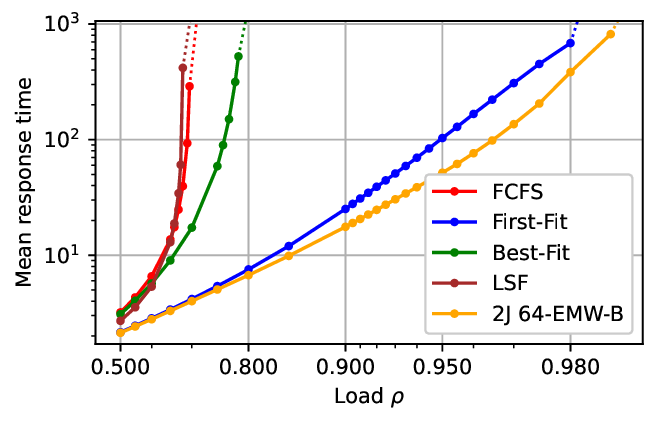}
    \caption{Log-log plot of mean response times v.s. load $\rho$, under theoretically-motivated and index-based policies when $V \sim \text{Uniform}((0,1])$. Simulation duration is $10^6$ jobs. }
    \label{fig: uniform-Comp}
\end{figure}

\subsubsection{Centrally Symmetric (Truncated Normal) Resource Requirement}
\label{sssec: truncated Normal}

We next consider a symmetric but non-uniform requirement distribution, to strengthen our empirical results in a setting of general symmetry. We further compare 2J $K$-EMW-B against $K$-MW-B, to show that choosing an appropriate efficient set helps to achieve both computational efficiency and a better mean response time.

The job resource requirement $V$ follows a truncated normal distribution $T\mathcal{N}(0.5,1)$ supported on $[0,1]$, i.e., $\mathcal{L}(V) = \mathcal{L}(Z | Z \in [0,1])$, where $Z \sim \mathcal{N}(0.5,1)$. Because this distribution is symmetric with mean $0.5$, we have $\lambda_V^* = 2$, and we know that the 2-Job $K$-Efficient MaxWeight (2J $K$-EMW) family is throughput-optimal, according to \cref{thm: 2JK-EMW Unif}.

Our observations from \cref{sssec: Uniform distribution} continue to hold in this centrally symmetric setting: In the absence of Backfilling, larger $K$ leads to worse mean response times until the arrival rate approaches the stability boundary, whereas larger $K$ combined with Backfilling results in uniformly better mean response times.

\begin{figure}[ht]
    \centering
    \begin{subfigure}[t]{0.49\linewidth}
        \centering
        \includegraphics[width=\linewidth]{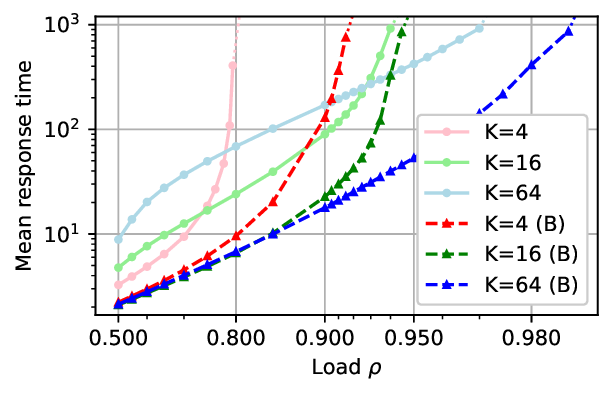}
        \caption{Mean response times under 2J $K$-EMW and 2J $K$-EMW-B. The curves for 2J $K$-EMW-B have labels ending with (B).   }
        \label{fig: TN(0.5,1)-EMW}
    \end{subfigure}
    \hfill
    \begin{subfigure}[t]{0.49\linewidth}
        \centering
        \includegraphics[width=\linewidth]{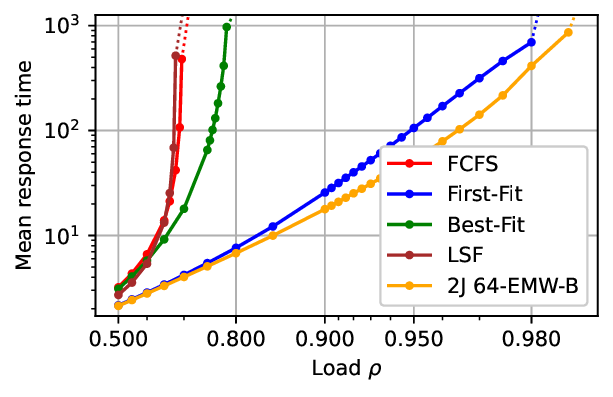}
        \caption{Mean response times under our theoretically-motivated policy and index-based policies.}
        \label{fig: TN-Comp}
    \end{subfigure}
    \caption{Log-log plots for mean response times v.s. load $\rho$ under 2J $K$-EMW with or without Backfilling, where $K=4, 16, 64$ (left), and under index-based policies and the selected theoretically-motivated policy 2J $64$-EMW-B (right).  $V \sim T\mathcal{N}(0.5,1)$. Simulation duration is $10^6$ jobs.}
    \label{fig: TN(0.5,1) Discretization}
\end{figure}

Next, in \cref{fig: TN-Comp}, we compare index-based policies with our theoretically-motivated policies. Consistent with previous results, 2J $64$-EMW-B achieves the largest empirically-stable regions, whereas FCFS, Best-Fit, LSF are indistinguishable from unstable under relatively small loads ($\rho \leq 0.8$). 

First-Fit, though leading to a worse mean response time than $K$-EMW-B, remains stable under high load, and we conjecture that it is stable in the single-resource setting when the resource requirement distribution is centrally symmetric. The reason for our conjecture is that First-Fit is proven to be throughput-optimal in \citet{coffman2001bandwidth}, in a setting close to our central symmetry setting. To be specific, \cite{coffman2001bandwidth} considers the discrete symmetric resource distribution setting in which complementary job types have equal arrival rates.

Finally, we compare the performance of our theoretically-motivated policies using the efficient set $E_K^{2J}$ and the full schedulable set $C_K$, to empirically show the improvement on mean response time when we choose an appropriate efficient set. According to \cref{rmk: E1 in E2 means E2MW also stable}, $K$-MW is stable when the stability condition for 2J $K$-EMW is satisfied, and therefore, the $K$-MW family is also throughput-optimal when $V\sim T\mathcal{N}(0.5,1)$. Hence, in \cref{fig: 8-Comp-MRT,fig: 16-Comp-MRT}, we fix $K=8$ and $K=16$, respectively, and compare the mean response times under $K$-MW, 2J $K$-EMW, $K$-MW-B and 2J $K$-EMW-B. We aim to examine, when stability is guaranteed, whether considering a larger set of service options can further reduce the mean response time.

\cref{fig: 8-Comp-MRT,fig: 16-Comp-MRT} show that for a fixed $K$, when the load is low, $K$-MW attains a better mean response time than 2J $K$-EMW. However, when the load is close to $1$, all policies have similar mean response times, and their empirically-stable regions are almost identical. 
This behavior is possibly because, in the low load case, by considering a larger service option set, we reduce the frequency that some availability is wasted, while when the load is high, there are always many jobs needing to be served, and the waste of availability becomes rare.

In contrast, when Backfilling is applied, 2J $K$-EMW-B achieves a better mean response time than $K$-MW-B. A possible explanation is that 2J $K$-EMW focuses on the most important subset of service options, while Backfilling avoids the failure to pair two jobs for joint service. 

Therefore, if given some $K$, $K$-EMW with some efficient set is sufficient to stabilize the system, we recommend using $K$-EMW-B with this efficient set, but not considering a larger set of service options, in order to achieve a larger empirically stable region and a better mean response time.

\begin{figure}[ht]
    \centering
    \begin{subfigure}[t]{0.48\linewidth}
        \centering
        \includegraphics[width=\linewidth]{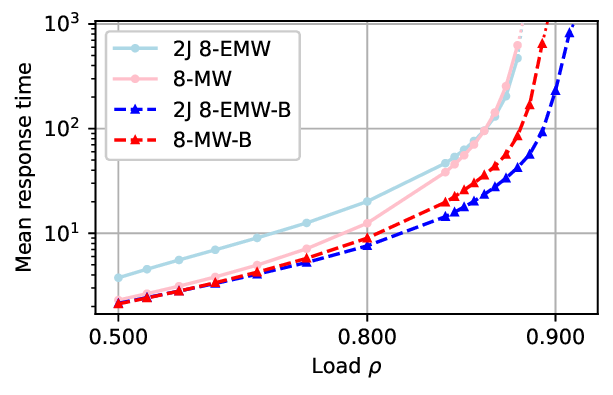}
        \caption{Mean response times v.s. load, $K=8$.}
        \label{fig: 8-Comp-MRT}
    \end{subfigure}
    \hfill
    \begin{subfigure}[t]{0.48\linewidth}
        \centering
        \includegraphics[width=\linewidth]{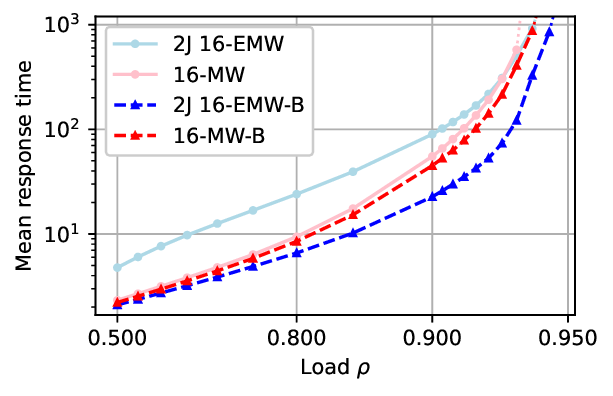}
        \caption{Mean response times v.s. load, $K=16$.}
        \label{fig: 16-Comp-MRT}
    \end{subfigure}
    \caption{Log-log plot of mean response times v.s. load $\rho$, under 2J $K$-EMW, $K$-MW, and their counterparts with Backfilling, when $V\sim T\mathcal{N}(0.5,1)$. Left: $K=8$; Right: $K=16$. Simulation duration is $10^6$ jobs.}
    \label{fig: Comp-MRT}
\end{figure}

\subsubsection{Weakly Decreasing (Bounded Lomax and Triangle) Resource Requirement}
\label{sssec: Bounded Lomax and Triangle (Weakly Decreasing) Resource Requirement}

We then study the performance of 2-Bucket $K$-discretized Efficient MaxWeight (2B $K$-EMW) and 2B $K$-EMW with Backfilling (2B $K$-EMW-B), and verify \cref{thm: 2DK-EMW decreasing}, which states the throughput-optimality of the 2B $K$-EMW family when job resource requirement distribution has decreasing density. Specifically, we consider a bounded Lomax distribution and a triangular distribution. For the bounded Lomax distribution $\text{BLomax}(2,1)$, the shape parameter and the scale parameter are set to $2$ and $1$, respectively, and the support of $V$ is $[0,1]$. The corresponding p.d.f. is
\[
f_V(v) = \frac{8}{3}(1+v)^{-3}, \quad 0 \leq v \leq 1.
\]
For the triangular distribution, we consider a linearly decreasing density over $[0,1]$, given by
\[
f_V(v) = 2 - 2v, \quad 0 \leq v \leq 1.
\]
Both distributions have mean $\E V = 1/3$, and therefore, by \cref{thm: 2DK-EMW decreasing}, $\lambda_V^* = 3$.

In \cref{fig: BLomax,fig: Tri(1)}, we mainly focus on comparing 2B $K$-EMW and 2B $K$-EMW-B with index-based policies. Similar to the observations in \cref{fig: TN-Comp}, our theoretically-motivated policies achieve the largest empirically-stable region, with a boundary close to $1$. The large empirically-stable region of 2B $64$-EMW also verifies \cref{thm: 2DK-EMW decreasing}.

In addition, we observe that First-Fit performs as stably as 2B $64$-EMW-B in each experiment, as in \cref{sssec: truncated Normal}. We therefore further conjecture that, under the assumption of a weakly decreasing resource requirement density, First-Fit is throughput-optimal, extending the results of \citet{coffman2001bandwidth} in the centrally symmetric setting.

\begin{figure}[ht]
    \centering
    \begin{subfigure}[t]{0.48\linewidth}
        \centering
        \includegraphics[width=\linewidth]{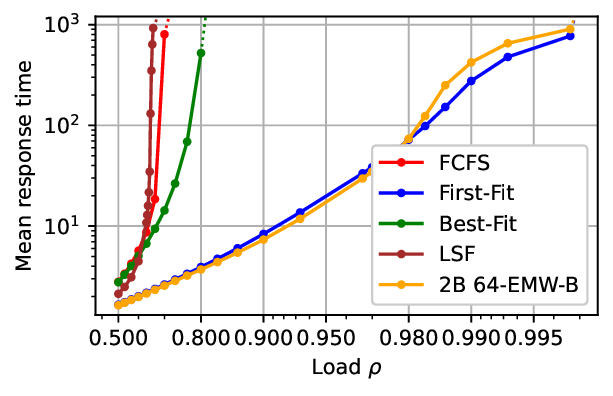}
        \caption{Mean response times when $V\sim \mathrm{BLomax}(2,1)$.}
        \label{fig: BLomax}
    \end{subfigure}
    \hfill
    \begin{subfigure}[t]{0.48\linewidth}
        \centering
        \includegraphics[width=\linewidth]{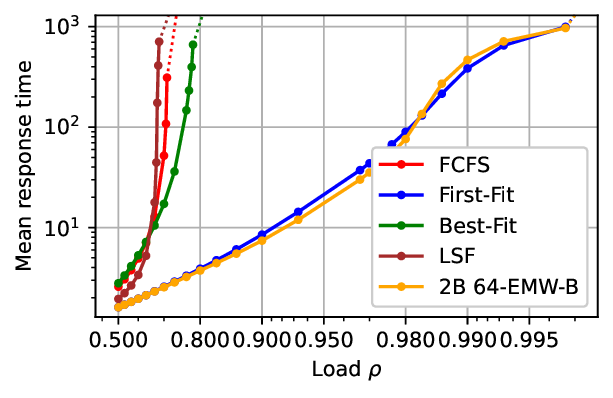}
        \caption{Mean response times when $V$ is triangularly distributed.}
        \label{fig: Tri(1)}
    \end{subfigure}
    \caption{Log-log plots of mean response times v.s. load $\rho$ under theoretically-motivated policies and index-based policies when $V \sim \text{BLomax}(2,1)$ or triangularly distributed, supported on $[0,1]$. Simulation duration is $10^6$ jobs.}
    \label{fig: Weakly Decreasing Discretization}
\end{figure}

\subsubsection{General Lipschitz: Triangular Resource Requirement}
\label{sssec: Sym Tri Resource Requirement}

We next consider a resource requirement distribution, the p.d.f. of which is Lipschitz continuous but neither centrally symmetric nor weakly decreasing, to empirically show that First-Fit is not always throughput-optimal, while our theoretically-motivated policies can achieve stability under higher load. We recall from \cref{thm: T-O of K-MW given Lipschitz condition} and \cref{cor: Backfilling Maintains Stability}, that the $K$-MW and the $K$-MW-B families are both throughput-optimal under this general Lipschitz setting.

We choose the symmetric triangular distribution $\text{SymmetricTri}(\frac{1}{4},\frac{1}{2})$ as our requirement distrbution, whose density is supported in $[\frac{1}{4},\frac{1}{2}]$ and has a shape of an isosceles triangle. Formally, we say $V\sim \text{SymmetricTri}(l,u)$ with lower limit $\ell$ and upper limit $u$ if its p.d.f. $f_V$ satisfies that
\[
f_V(v)=
\begin{cases}
\dfrac{4(v-\ell)}{(u-\ell)^2}, & \ell \leq v \leq \dfrac{\ell+u}{2},\\
\dfrac{4(u-v)}{(u-\ell)^2}, & \dfrac{\ell+u}{2} < v \leq u,\\
0, & \text{otherwise}.
\end{cases}
\]
Because the server can always serve at least two jobs simultaneously, $\lambda^*_{V}\geq 2$. According to \cref{lem: MRJ stability region subset}, we can upper bound $\lambda^*_{V}\leq \frac{1}{\E V} = \frac{8}{3}$. However, the exact stability boundary $\lambda^*_{V}$ is unknown, and hence we do not use the load here. We choose $30$-discretized MaxWeight with backfilling ($30$-MW-B), because $30$ has more different divisors. Empirically, $K$-MW and $K$-MW-B tend to achieve a better performance if $K$ has more different divisors. 

In addition, we evaluate the Pairwise Extreme-vertices $K$-Efficient MaxWeight with Backfilling policy (XP $K$-EMW-B), which serves as a tractable proxy for the Extreme-vertices policy (X $K$-EMW-B), while remaining far more efficient than full MaxWeight.  By \cref{thm: vertices KEMW and KEnMSR} and \cref{rmk: XP K-EMW and XP K-EMW are TO}, both the X $K$-EMW-B and XP $K$-EMW-B families are throughput-optimal. In addition, for $K=30$, XP $30$-EMW-B uses only 980 service options, whereas $30$-MW-B uses 5,604 service options.

\cref{fig: Iso} compares our policies with several index-based policies.  We observe that $30$-MW-B and XP $30$-EMW-B have indistinguishable mean response times for each arrival rate, even though XP $30$-EMW-B uses more than five times fewer service options. 
Moreover, both our theoretically-motivated policies achieve the largest empirically-stable region and the uniformly best mean response time, while all index-based policies become unstable under high load.  These simulation results demonstrate the robustness of $K$-MW-B and XP $K$-EMW-B, which consistently achieve the largest stability region across different resource requirement distribution settings. Moreover, replacing the full service-option set $C_K$ with the smaller efficient set $E^{\mathrm{XP}}_K$ can substantially improve computational efficiency while preserving nearly the identical empirical performance.

\begin{figure}[ht]
    \centering
    \includegraphics[width=0.95\linewidth]{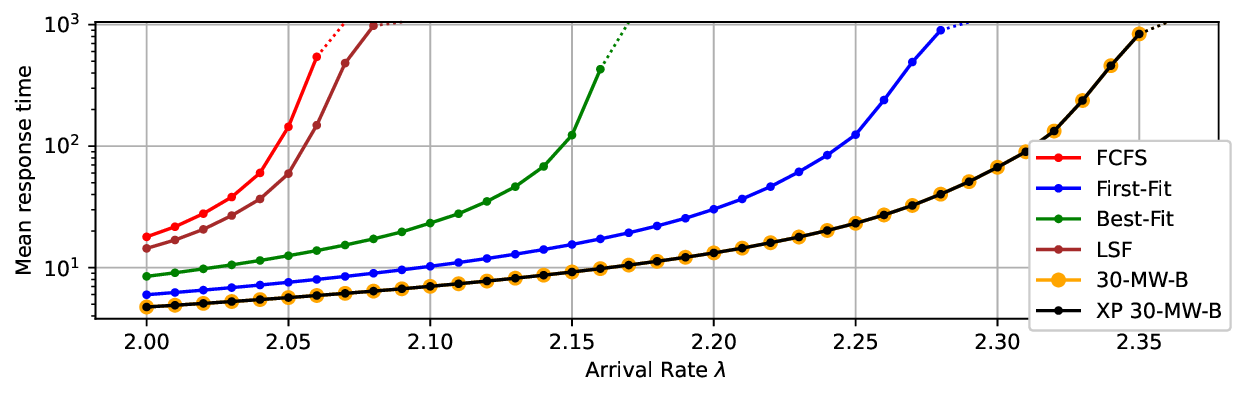}
    \caption{The plot of mean response times under $30$-MW-B, XP $30$-MW-B, and index-based policies v.s. the arrival rate, where when $V \sim \text{SymmetricTri}(\frac{1}{4},\frac{1}{2})$. Simulation duration is $10^6$ jobs.}
    \label{fig: Iso}
\end{figure}

\subsection{Our Theoretically-motivated Policies v.s. index-based Policies: Borg Traces}
\label{ssec: Discretized MaxWeight-typed Policies v.s. index-based Policies: Borg Traces}

We next evaluate the practical performance of our policy using the CPU and RAM requirements from a Google Borg trace \cite{yildiz2024data}\footnote{Data available at the following repository: \url{https://github.com/MertYILDIZ19/Google_cluster_usage_traces_v3_single_cell}}, and show the state-of-the-art performance of our theoretically-motivated policies.

In our first simulation experiment, we consider an MSJ model in which the job arrival order follows the order observed in the Google Borg traces, and the resource requirement of each job is given by its rescaled CPU demand. We rescale the CPU requirements because the system capacity is unknown and the empirical distribution is highly skewed: Only a very small fraction of jobs request relatively large amounts of CPU, whereas the majority require extremely small amounts. As a result, directly using the raw data would require an excessively fine discretization, with most buckets receiving essentially zero arrivals. To address these issues, we discard the largest $10\%$ of the observations and normalize the remaining requirements. In the second simulation experiment, we rescale the RAM requirements in the same way.

We assign each job an independent $Exp(1)$ service time. Arrival times are modeled by a Poisson process $PP(\lambda)$. However, because the resource requirements are not i.i.d. and the resulting system is non-stationary, the stability threshold $\lambda_V^*$ does not apply, and hence we do not define the load in this section.

Because the two requirement distributions are approximately weakly decreasing, the 2B $K$-EMW-B policy is a natural fit. We set $K=64$ and compare 2B $64$-EMW-B with index-based policies. In addition, the inherently discrete value of requirements in the trace motivates us to include the class-based policy in our comparison. However, in these 1 million data points, the number of unique RAM requirements and unique CPU requirements are 3,668 and 536,229, respectively. This large number of classes makes it infeasible to run naive MaxWeight in our simulation. Instead, we propose the following class-based policy to approximate naive MaxWeight, called Pseudo MaxWeight (Pseudo-MW). Recall that naive MaxWeight performs identically to LSF when the resource requirements are continuously distributed. Similar to LSF, Pseudo-MW packs the jobs in increasing order of the ratio of resource requirement and the current number of jobs of that type. In addition to using Pseudo-MW to approximate naive MaxWeight, we further examine whether this policy achieves a similar performance to LSF.

The simulation results in Figures \cref{fig: Borg CPU,fig: Borg Memory} both show that 2B $64$-EMW-B significantly outperforms the index-based policies and Pseudo-MW, achieving substantially smaller mean response times. In addition, Pseudo-MW performs almost identically to LSF, which implies that when the class size is extremely large, class-based policy is not distinguishable from the index-based policies. These results show that our continuously-motivated policies can still exhibit superior performance in more realistic settings, where the resource requirements are not i.i.d., and are discrete with large support.

\begin{figure}[ht]
    \centering
    \begin{subfigure}[t]{0.48\linewidth}
        \centering
        \includegraphics[width=\linewidth]{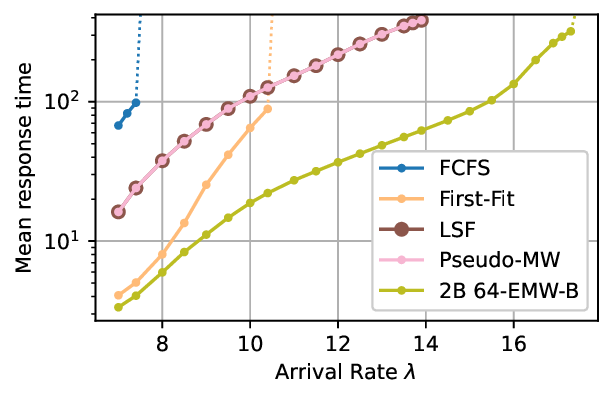}
        \caption{Mean response times (CPU data).}
        \label{fig: Borg CPU}
    \end{subfigure}
    \hfill
    \begin{subfigure}[t]{0.48\linewidth}
        \centering
        \includegraphics[width=\linewidth]{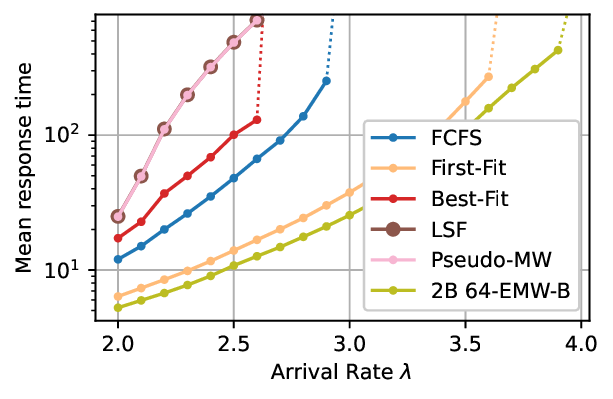}
        \caption{Mean response times (RAM data)}
        \label{fig: Borg Memory}
    \end{subfigure}
    \caption{The plot of mean response times under 2B $64$-EMW-B, Pseudo-MW, and index-based policies v.s. the load, where job resource requirements correspond to the first one million CPU requirements and the first one million RAM requirements from the public Google Borg trace \cite{yildiz2024data}, respectively. Each requirement's data is normalized by the $90\%$ quantile of the population data, and the original data greater than the $90\%$ quantile are excluded. Each simulation duration is the first $10^6$ jobs in the trace.
    In the first experiment corresponding to \cref{fig: Borg CPU}, Best-Fit is indistinguishable from unstable for any arrival rate greater than $7$, and we do not include it in the plot. In both experiments, Pseudo-MW achieves almost identical mean response times to LSF, and their curves are overlapped in both \cref{fig: Borg CPU,fig: Borg Memory}.}
    \label{fig: Borg experiment}
\end{figure}

\section{Conclusion}
\label{sec: conclusion}

We introduce the continuous-resource multiresource-job (MRJ) system.
We show the first throughput-optimality results in this setting, for the  $\K$-discretized MaxWeight family, and a related nonpreemptive policy family.
Moreover, we introduce three efficient variants, which achieve the same throughput optimality with more efficient computation, under distributional assumptions.

We demonstrate via simulation the advantages of our policies when combined with Backfilling, relative to existing index-based policies. Most importantly, we evaluate our policies using Google Borg traces as job requirements. Despite the fact that these traces are not independently distributed, our simulation shows that our 2-Bucket Efficient MaxWeight achieves the best mean response time uniformly over all arrival rates.

A promising direction for future work lies in analyzing a MaxWeight variant where the discretization parameter is chosen dynamically based on the queue length, as well as nonuniform discretization. Another future direction in the single-resource setting is to prove that the throughput-optimality of First-Fit and Best-Fit, when the resource requirement distribution has symmetric or decreasing density \cite{coffman2001bandwidth}.




\backmatter





\bmhead{Acknowledgments}
The first author would like to thank Yuxi Chen, Zhongrui Chen, Cameron Curtis, Ziyuan Wang, and Bryan Yuen for helpful discussions and comments. The authors also thank Alexander Stolyar for the discussion on the centrally symmetric resource distribution setting and for pointing out the related results in \citet{coffman2001bandwidth}.




\section*{Declarations}


\begin{itemize}
\item Funding: This work is supported by a Northwestern IEMS Startup Grant. 
\item Competing interests: The authors have no relevant financial or non-financial interests to disclose.
\item Data and code availability: Our code is available at the following repository: \url{https://github.com/isaacg1/buckets}. For the Google Borg trace data used in our simulation, we refer the reader to the following repository: \url{https://github.com/MertYILDIZ19/Google_cluster_usage_traces_v3_single_cell}.
\item Author contribution: Theoretical analysis and manuscript writing are conducted by Heyuan Yao and Izzy Grosof. Simulation software is created by Willow Kowalik and Izzy Grosof.
\end{itemize}







\bibliography{references}


\begin{appendices}

\section{List of Notation}
\label{sec: List of Notations}
General Setting:

\begin{table}[ht]
\centering
\begin{tabular}{cl}
\hline
Symbol                            & \multicolumn{1}{c}{Description}                                                                 \\ \hline
$\lambda$                         & Arrival rate                                                               \\
$d$                               & The number of resources in the MRJ setting                                    \\
$V$                               & Resource requirement (joint) distribution                                \\
$f_V$                             & Probability density  function of $V$                                      \\
$\eta^A$                           & Arrival measure with density  $\lambda f_V$  \eqref{equ: def of mu^A}                            \\ \hline
\end{tabular}
\end{table}

\newpage
Multiserver (MSJ) setting under $K$-discretization system:

\begin{table}[ht]
\centering
\begin{tabular}{cl}
\hline
Symbol                            & \multicolumn{1}{c}{Description}                                                                 \\ \hline
$K$                               & Discretization parameter              \\
$k$                               & Job type, $1\leq k \leq K$                                  \\
$I_k$                             & The $k$-th bucket $(\frac{k-1}{K},\frac{k}{K} ]$   \eqref{equ: K-partition}                      \\
$\eta^A_K$                         & $K$-discretized arrival measure   on $[K]$                                 \\
$\Lambda_K   \in \R^K$            & Vector of arrival rates of   type-$1$ to type-$K$ jobs                     \\
$C_K\subset   \N^K$               & Set of schedulable service   options \eqref{equ: C_K a-dim}                                       \\
$M\in   C_K$                      & Service option                                                             \\
$M^{(k)}$                         & The number of type-$k$ job(s)   $M$ serves                                 \\
$E_K \subseteq C_K$               & Efficient set (unspecific)                                                 \\
$E^{X}_K$                         & The Extreme-vertices efficient set                                       \\
$E^{XP}_K$                         & The Pairwise Extreme-vertices efficient set ($K$ is even)                                       \\
$E^{2B}_K$                        & The 2-Bucket efficient set                                                 \\
$E^{2J}_K$                        & The 2-Job efficient set                                                    \\
$\beta_K$                           & A probability distribution over $C_K$ (or some $E_K$)                    \\
$\eta^S_K$                         & $K$-discretized service measure on $[K]$, associated with $\beta_\K$        \\ \hline
\end{tabular}
\end{table}

Multiresource (MRJ)  setting under $\K$-discretization system:
\begin{table}[ht]
\centering
\begin{tabular}{cl}
\hline
Symbol                            & \multicolumn{1}{c}{Description}                                                                 \\ \hline
$\K \in \N^d$                     & Discretization parameter         \\
$\vi$                             & Job type                                             \\
$I_\vi$                           & The $\vi$-th bucket \eqref{equ: partition in d-dim}                                                    \\
$\Gamma_\K$                       & Job type set $\{\vec 1 \preceq \vi  \preceq \K\}$                        \\
$\eta^A_\K$                        & $\K$-discretized arrival measure   on $\Gamma_\K$ \eqref{equ： discretized arrival measure}                       \\
$\Lambda_\K  \in \R^{\Gamma_\K}$ & Vector  of arrival rates of each type of job                                \\
$C_\K\subset   \N^{\Gamma_\K}$    & Set of schedulable service options \eqref{equ: C_K d-dim}                 \\
$M\in C_\K$                     & Service option                                                             \\
$M^{(\vi)}$                       & The number of type-$\vi$ job(s)   $M$ serves                               \\
$E_K \subseteq C_K$               & Efficient set (unspecific)                                                 \\
$E^{X}_\K$                        & The Extreme-vertices efficient set                                       \\
$E^{2J}_\K$                       & The 2-Job efficient set                                                    \\
$\beta_\K$                           & A probability distribution over $C_\K$ (or some $E_\K$)                  \\
$\eta^S_\K$                        & $\K$-discretized service measure on $\Gamma_\K$, associated with $\beta_\K$ \eqref{equ: K service measure by M_k}\\ \hline
\end{tabular}
\end{table}

\section{Definitions of Some Index-Based Policies}
\label{sec: Definitions of Some Index-Based Policies}

In this section, we define First-Fit, Best-Fit, and Least-Server-First (LSF), which are the index-based policies used in \cref{sec: empirical results} as simulation benchmarks for our policies.

We first describe First-Fit in the MSJ setting. Whenever the system state changes, First-Fit scans the jobs in the queue according to their arrival order and greedily constructs a feasible set of jobs for joint service. To be specific, suppose that the policy has scanned the first $i$ jobs in this order, and let $x$ denote the total resource requirement of the jobs that have already been selected for service. Because the server has total capacity $1$, the remaining available capacity is $1-x$. If the $i+1$-th scanned job has a requirement $r \leq 1-x$, then this job is selected for service, the total requirement is updated to $x+r$, and the policy moves to scan the next job. Otherwise, the job is skipped, and the policy proceeds to scan the next job without changing the selected set. First-Fit continues this procedure until all jobs in the queue have been scanned. In addition, First-Fit can be naturally extended to the MRJ setting. In this case, each job has a resource-requirement vector, and the server has unit capacity for each resource type. A scanned job is selected for service whenever its requirement for every resource type does not exceed the corresponding remaining capacity.

Likewise, in the MSJ setting, Best-Fit and LSF follow the same greedy packing rule but use different scanning orders. Best-Fit attempts to pack larger jobs first, while LSF prioritizes smaller jobs. Formally speaking, Best-Fit scans the jobs in decreasing order of their resource requirements, whereas LSF scans the jobs in increasing order of their resource requirements. We note that LSF does not need to scan the entire queue. Once the remaining capacity is smaller than the requirement of the current scanned job, any subsequently scanned job cannot fit into the remaining capacity. Hence, LSF can terminate the scan at the first time when a job fails to fit.

\section{Upper bound for the Supremum of Lipschitz Density Function}
\label{appsec: Upper bound for the Supremum of Lipschitz Density Function}
In this section, we derive an upper bound $L(C,d)$ on the probability density function $f$ on $[0,1]^d$ where we assume the Lipschitz constant of $f$ is $C$. Formally speaking, we bound
\[ \sup \left \{ \norm{f}_\infty: f:[0,1]^d\rightarrow \R_+,\; \norm{f}_{Lip} \leq C, \int_{[0,1]^d} f = 1 \right \}.\]

\begin{lemma}[Supremum Bound of Lipschitz Continuous p.d.f. on Unit Cube]
    \label{lem: Upper bound for the Supremum of Lipschitz Density Function}
    For any probability density function $f$ defined on $[0,1]^d$, if $f$ is Lipschitz continuous with Lipschitz constant $C$, then $\norm{f}_\infty \leq L(C,d)$, where 
    \begin{equation}
        \label{equ: choice of M}
        L(C,d): = \lbr\frac{(d+1) \pi^{\frac{d}{2}}(\max \{C, C^*\})^d}{2^d \Gamma(\frac{d}{2}+1)} \rbr^{\frac{1}{d+1}}, \quad \text{ where } C^* = \frac{(d+1) \pi^{\frac{d}{2}}}{2^d \Gamma(\frac{d}{2}+1)}.
    \end{equation}
\end{lemma}

\begin{proof}
    For any dimension $d$, when $C$ is large enough, the upper bound of $\sup \norm{f}_\infty$ is achieved by a function with maximum at $\vec 0$ and directional derivative $-C$ in all directions.    
    
    In the single-resource setting, the integral of $f$ is the area of a right triangle with side lengths $L$ and $L/C$, i.e., $\frac{L^2}{2C} = 1 \Rightarrow L = (2C)^{\frac{1}{2}}$. Thus $L(C,1) \leq \sqrt{2C}$, given large enough $C$. Specifically, the above holds when $L/C \leq 1$, i.e, $C\geq 2$. Thus $L(C,1) \leq \max (\sqrt{2C},2)$.

    In $d$-dimensional case, the integral of $f$ is the same as $2^{-d}$ the volume of a $d+1$-dimensional circular cone with radius $L/C$ and height $L$, when $d \geq 2$. By applying the surface area formula in 
    $d$-dimensional space, $A(d,r) = \frac{d \pi^{\frac{d}{2}}}{\Gamma(\frac{d}{2}+1)}r^{d-1}$, where $r$ denotes the radius and $\Gamma(\cdot)$ denotes the Gamma function, we have that 
    \begin{align*}
        1 & = \int_0^{L/C} \frac{d \pi^{\frac{d}{2}}}{2^d \Gamma(\frac{d}{2}+1)}r^{d-1} (L-Cr) dr =  \frac{d \pi^{\frac{d}{2}}}{2^d \Gamma(\frac{d}{2}+1)} \frac{L^{d+1}}{d (d+1) C^d } \\
        & \Leftrightarrow L = \lbr\frac{(d+1) \pi^{\frac{d}{2}}C^d}{2^d \Gamma(\frac{d}{2}+1)} \rbr^{\frac{1}{d+1}}.
    \end{align*}
    Note that therefore the discussion above requires $C \geq \frac{(d+1) \pi^{\frac{d}{2}}}{2^d \Gamma(\frac{d}{2}+1)}=: C^*$. Note that $L(C^*,d) = C^*$.
    
    We conclude that $L(C,d)$ can be chosen as $\lbr\frac{(d+1) \pi^{\frac{d}{2}}C^d}{2^d \Gamma(\frac{d}{2}+1)} \rbr^{\frac{1}{d+1}}$, when $C\geq C^*$. When $C<C^*$, we can trivially choose the upper bound $L(C,d) = L(C^*,d) = C^*$.
\end{proof}




\end{appendices}


\end{document}